\documentclass[10pt,journal,compsoc]{IEEEtran}
\newcommand{\subparagraph}{}
\usepackage{titlesec}
\titleformat{\paragraph}
{\normalfont\normalsize\bfseries\itshape}{\theparagraph}{1em}{}
\titlespacing*{\paragraph}
{0pt}{3.25ex plus 1ex minus .2ex}{1.5ex plus .2ex}

\usepackage{empheq}
\usepackage{arydshln}

\usepackage{enumitem}

\usepackage{blindtext}

\usepackage{comment}

\ifCLASSOPTIONcompsoc
  \usepackage[nocompress]{cite}
\else
  \usepackage{cite}
\fi
   \usepackage[pdftex]{graphicx}

\usepackage{tikz}
\newcommand\tikzmark[1]{\tikz[remember picture,overlay]\coordinate (#1);}

\usepackage{amsmath}
\usepackage{amssymb}
\usepackage{dsfont}
\newcommand{\mathbbm}{\mathds}

\usepackage{amsthm}

\ifCLASSOPTIONcompsoc
  \usepackage[caption=false,font=footnotesize,labelfont=sf,textfont=sf]{subfig}
\else
  \usepackage[caption=false,font=footnotesize]{subfig}
\fi

\numberwithin{equation}{section}

\begin{document}

\newtheorem{theorem}{Theorem}[section]
\newtheorem{lemma}[theorem]{Lemma}
\newtheorem{proposition}[theorem]{Proposition}
\newtheorem{corollary}[theorem]{Corollary}

\renewenvironment{proof}{\noindent{\bfseries Proof.}}{\qed}
\renewcommand{\qedsymbol}{$\blacksquare$}

%
\title{Analysis of Exact and Approximated Epidemic Models over Complex Networks}
%
%
%
%

\author{Navid~Azizan~Ruhi,~\IEEEmembership{Student~Member,~IEEE,}
        Hyoung~Jun~Ahn,~\IEEEmembership{Student~Member,~IEEE,}
        and~Babak~Hassibi,~\IEEEmembership{Member,~IEEE}
\IEEEcompsocitemizethanks{\IEEEcompsocthanksitem N. Azizan Ruhi and H. J. Ahn are with the Department
of Computing and Mathematical Sciences, California Institute of Technology, Pasadena,
CA, 91125.\protect\\
E-mail: \{azizan, ctznahj\}@caltech.edu
\IEEEcompsocthanksitem B. Hassibi is with the Department
of Electrical Engineering, California Institute of Technology, Pasadena,
CA, 91125.\protect\\
E-mail: hassibi@systems.caltech.edu}
}

%
%

\markboth{}%
{Azizan Ruhi \MakeLowercase{\textit{et al.}}: Analysis of Exact and Approximated Epidemic Models over Complex Networks}
%



\IEEEtitleabstractindextext{%
\begin{abstract}
We study the spread of discrete-time epidemics over arbitrary networks for well-known propagation models, namely SIS (susceptible-infected-susceptible), SIR (susceptible-infected-recovered), SIRS (susceptible-infected-recovered-susceptible) and SIV (susceptible-infected-vaccinated). Such epidemics are described by $2^n$- or $3^n$-state Markov chains.
Ostensibly, because analyzing such Markov chains is too complicated, their $O(n)$-dimensional nonlinear ``mean-field'' approximation, and its linearization, are often studied instead.
We provide a complete global analysis of the epidemic dynamics of the nonlinear mean-field approximation. In particular, we show that depending on the largest eigenvalue of the underlying graph adjacency matrix and the rates of infection, recovery, and vaccination, the global dynamics takes on one of two forms: either the epidemic dies out, or it converges to another unique fixed point (the so-called endemic state where a constant fraction of the nodes remain infected).
A similar result has also been shown in the continuous-time case.
We tie in these results with the ``true'' underlying Markov chain model by showing that the linear model is the tightest upper-bound on the true probabilities of infection that involves only marginals, and that, even though the nonlinear model is not an upper-bound on the true probabilities in general, it does provide an upper-bound on the probability of the chain not being absorbed.
As a consequence, we also show that when the disease-free fixed point is globally stable for the mean-field model, the Markov chain has an $O(\log n)$ mixing time, which means the epidemic dies out quickly. We compare and summarize the results on different propagation models.
\end{abstract}

}

\maketitle

\IEEEdisplaynontitleabstractindextext

%
\IEEEpeerreviewmaketitle

\IEEEraisesectionheading{\section{Introduction}\label{sec:introduction}}
\IEEEPARstart{E}{pidemic} models have been extensively studied since a first mathematical formulation was introduced in 1927 by Kermack and McKendrick \cite{kermack1927contribution}. Though initially proposed to understand the spread of contagious diseases \cite{bailey1975mathematical}, the study of epidemics applies to many other areas, such as network security \cite{alpcan2010network,acemoglu2013network}, viral advertising \cite{phelps2004viral,richardson2002mining}, and information propagation \cite{jacquet2010information,cha2009measurement}. Questions of interest include the existence of fixed-points, stability (does the epidemic die out), transient behavior, the cost of an epidemic \cite{LizMinimizing,bose2013cost}, how best to control an epidemic \cite{drakopoulos2014efficient,nowzari2015analysis}, etc.

We analyze the spread of epidemics over arbitrary networks for most well-known propagation models in the literature, including SIS (susceptible-infected-susceptible), SIR (susceptible-infected-recovered), SIRS (susceptible-infected-recovered-susceptible), and SIV (susceptible-infected-vaccinated). In the basic SIS model, each node in the network is in one of two different states: susceptible (healthy) or infected. A healthy node has a chance of getting infected if it has infected neighbors in the network. The probability of getting infected increases as the number of infected neighbors increases. An infected node also has a chance of recovering after which it still has a chance of getting infected by its neighbors. Flu is an example of this model. SIR and SIRS models have an extra recovered state, which corresponds to the nodes that have recovered from the disease and are not susceptible to it. Mumps and Pertussis respectively are examples of SIR and SIRS epidemics. Additionally, in SIV models, there is a random vaccination (either permanent or temporary) which permits direct transition from the susceptible state to the recovered (vaccinated) one.

Considering even the SIS case in its entirety, for a network with $n$ nodes, this yields a Markov chain with $2^n$ states, sometimes called the exact or ``stochastic'' model. This is a discrete-space model, as there are two possible states of ``0'' and ``1'' for healthy and infected.
Ostensibly, because analyzing this Markov chain is too complicated, various $n$-dimensional linear and non-linear approximations have been proposed. The most common of these is the $n$-dimensional non-linear mean-field approximation, and its corresponding linearization about the disease-free fixed point, which are often referred to as ``deterministic'' models. Indeed these are  continuous-space models, that take real numbers between 0 and 1, which can be understood as the marginal probability for being infected (or the infected fraction of the $i$-th subpopulation).

We provide a complete global analysis of the dynamics of the nonlinear model. In particular, we show that depending on the largest eigenvalue of the underlying graph adjacency matrix and the rates of infection, recovery, and vaccination, the global dynamics takes on one of two forms: either the epidemic dies out (disease-free fixed point), or it converges to another unique fixed point where a constant fraction of the nodes remain infected (endemic state).
Furthermore, we tie in the approximated models with the \emph{true} underlying Markov chain model. We prove that the linear model provides an upper bound on the marginal probabilities of infection, and this is the \emph{tightest upper-bound using the marginals only}. We show that, even though the nonlinear model is not an upper-bound on the true probabilities in general, it does provide an upper-bound on the probability of the chain not being absorbed (some nodes being infected). As a consequence of these results, we show that when the $O(n)$-dimensional approximate models are stable to the disease-free fixed point, the Markov chain has a mixing time of $O(\log n)$, which means the epidemic dies out fast in the true model as well.

The study of continuous-time and discrete-time epidemic models are two parallel bodies of work, and interesting results have been shown in both cases by different groups of researchers, e.g. \cite{Ganeshnetworktopology,van2014upper,fall2007epidemiological,shuai2013global,li2014analysis,nowzari2014stability,draief2006thresholds} in the continuous-time and \cite{arenas2010discrete,WangEpidemic,chakrabarti2008epidemic,ahn2013global,prakash2012threshold,ahn2014mixing} in the discrete-time case. Depending on the application in hand it may make more sense to use one class or the other. This paper focuses on discrete-time models, and we provide a unified analysis of exact and approximated models and the connections between them. We spell out our contributions with respect to what is known in both the discrete-time and the continuous-time literature, below.\\
The following results were not known in either of the discrete- or continuous-time literature:
\begin{enumerate}
\item We show that the linear model is the tightest upper-bound with the marginals only on the exact probabilities of infection.
\item We show that, even though the nonlinear model is not an upper-bound on the exact probabilities in general, it does provide an upper-bound on the probability of the chain not being absorbed.
\item Although the logarithmic time-to-extinction of the epidemic under the threshold was known for the SIS model in the continuous-time case (Ganesh et al. \cite{Ganeshnetworktopology}), this result had not been shown for other well-known propagation models (e.g. SIRS, SIV, etc.) in either discrete-time or continuous-time.
\end{enumerate}
In addition to the above, we complement the discrete-time literature by showing the following results that were recently shown in the continuous-time case \cite{shuai2013global,khanafer2014stability,fall2007epidemiological} but not in the discrete-time one. 
\begin{enumerate}
\item In discrete-time mean-field approximated models, the stability of the disease-free fixed point under the threshold had been shown for SIS and many more complicated propagation models. However, the existence and stability of a unique endemic equilibrium above the threshold had not been shown for any discrete-time model, before this work.
\item Contrary to the continuous-time literature, the stability results shown for discrete-time approximated models are typically “local.” But we show “global stability” results, which are counterparts of the continuous-time case.
\end{enumerate}

Sections 2, 3, and 4, are devoted to SIS, SIRS, and SIV epidemic models, respectively. Starting from SIS epidemics, we describe the exact Markov chain model, the nonlinear epidemic map, and the linear model. In the analysis of the nonlinear model, we first describe the case where the epidemic dies out. Then we analyze the second case where the all-healthy fixed point is not stable, and show the existence and uniqueness of a second fixed point, and its global stability. Returning back to the exact Markov chain model, we establish the connection between that and the approximated models. We define a partial order which makes the transition matrix of the MC an order-preserving map, and helps us to establish the relation. We further generalize the model by allowing each node to have its own recovery and infection rates. We discuss variations of the models depending on the effect of simultaneous recovery and infection, as well as the efficacy of the vaccination.
Simulation results for all the models are provided in Section 5, which support the results proved throughout the paper.
We finally summarize the results, compare them, and conclude in Section 6.
To avoid confusion and facilitate reading, we use boxes for the main equations describing the models in each section. The proofs are postponed to the appendix. The current paper combines and expands the results that first appeared in \cite{ahn2013global},\cite{ahn2014mixing},\cite{ruhi2015sirs}.
\section{SIS Epidemics}
\begin{figure}[hb]
  \centering
    \includegraphics[width=0.4\columnwidth]{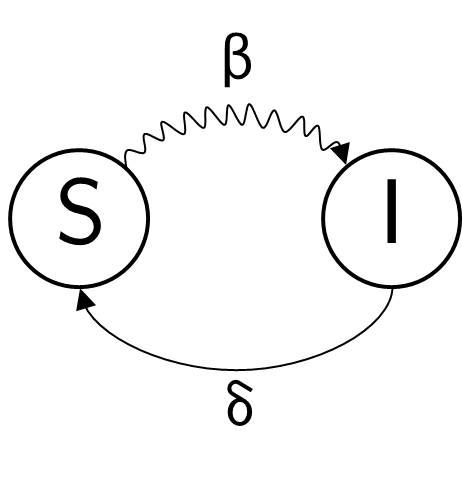}
      \caption{State diagram of a single node in the SIS model, and the transition rates. Wavy arrow represents exogenous (network-based) transition. $S$ is healthy or susceptible, $I$ is infected.}
      \label{fig:SIS_model}
\end{figure}
\subsection{Model Description}
\subsubsection{Exact Markov Chain Model}
For a given connected network $G$ with $n$ nodes, let $N_i$ be the neighborhood of node $i$. Let $A$ be the adjacency matrix of $G$. Each node can be in a state of health, represented by ``0'', or a state of infection, represented by ``1''. Consequently, $\xi(t)=(\xi_1(t), \cdots, \xi_n(t)) \in \{0,1 \}^n$ is a binary n-tuple where each of its entries represents the state of each node at time $t$. i.e. $i$ is infected if $\xi_i(t) =1$ and it is healthy if $\xi_i(t)=0$.

We assume that probability of infection of each node given the current state $\xi(t)$ is independent. In other words, for any two state vectors $X,Y \in \{0,1\}^n$,
\begin{equation}
\mathbb{P}(\xi(t+1)=Y|\xi(t)=X) = \prod_{i=1}^n \mathbb{P}(\xi_i(t+1)=Y_i|\xi(t)=X) \label{MC0}
\end{equation}

A healthy node remains healthy if all its neighbors are healthy. A healthy node can receive infection from any of its infected neighbors independently with probability $\beta$. An infected node becomes healthy if it is recovered from the disease with probability $\delta$ while not getting infected by any of its neighbors. To summarize this,
\begin{empheq}[box=\mbox]{align}
\tikzmark{A}
&\mathbb{P}(\xi_i(t+1)=Y_i|\xi(t)=X) \nonumber \\
&= \left\{
\begin{array}{rl}
(1-\beta)^{m_i} & \text{if } (X_i,Y_i)=(0,0),  |N_i \cap \mathbb{S}(X)| = m_i,\\
1- (1-\beta)^{m_i} & \text{if } (X_i,Y_i)=(0,1),  |N_i \cap \mathbb{S}(X)| = m_i,\\
\delta(1-\beta)^{m_i} & \text{if } (X_i,Y_i)=(1,0),  |N_i \cap \mathbb{S}(X)| = m_i,\\
1-\delta(1-\beta)^{m_i} & \text{if } (X_i,Y_i)=(1,1),  |N_i \cap \mathbb{S}(X)| = m_i.
\end{array} \right.\tikzmark{B} \label{Eq:noimmune}
\end{empheq}
\begin{tikzpicture}[remember picture,overlay]
\draw ([shift={(-.2em,2.5ex)}]A) rectangle ([shift={(-0.5em,-6.0ex)}]B);
\end{tikzpicture}
where $\mathbb{S}(X)$ is the support of $X \in \{0,1\}^n$, i.e. $\mathbb{S}(X) = \{ i : X_i = 1 \}$.

Let $S$ be the transition matrix of this Markov Chain, $S_{X,Y}=\mathbb{P}(\xi(t+1)=Y|\xi(t)=X)$. We assume that the Markov chain is time-homogeneous and write $S_{X,Y}=\mathbb{P}(Y|X)$ for simplicity.

The Markov chain has a unique stationary distribution, which is the state where all the nodes in the
network are healthy with probability $1$. If all the nodes are healthy, no node will be exposed to disease, and therefore they will always stay healthy. Therefore the probability distribution on the states $\{0,1\}^n$, goes to the all-healthy state as time progresses. In other words, the disease will die out if we wait long enough. However, this result is not practical since it may take a very long time especially if the mixing time of the Markov chain is exponentially large. It is difficult to analyze the dynamics of the Markov chain as the number of nodes increases.

Comparing the discrete-time Markov chain model to the continuous-time Markov chain model described in \cite{Ganeshnetworktopology}, continuous-time Markov chain model allows only one flip of each node's epidemic state at each moment. However, the discrete-time model allows change of epidemic states for more than one node at the same time. The reason being that change of epidemic state for two or more nodes can occur at same time interval, even though they do not happen at the same moment. The transition matrix of the embedded Markov chain of continuous-time model has nonzero entries only where the Hamming distance of row coordinate and column coordinate is 1. In other words, the number of different digits for $X,Y \in \{0,1\}^n$ should be 1 in order for the entry of the $X$-th row and the $Y$-th column to be nonzero. However, the transition matrix of the discrete-time Markov chain model can have nonzero entries everywhere (except the row of the absorbing state). 

Denote $I(t)$ as the set of infected nodes at time $t$. Define $p_i(t)$ as the marginal probability that node $i$ is infected at time $t$, i.e. $p_i(t)=\mathbb{P}(\xi_i(t)=1)$.
\begin{align}
p_i(&t+1)=\mathbb{P}(\xi_i(t+1)=1 | \xi_i(t)=1) \mathbb{P}(\xi_i(t)=1)\notag\\ 
&\hspace{35pt}+\mathbb{P}(\xi_i(t+1)=1 | \xi_i(t)=0) \mathbb{P}(\xi_i(t)=0)
\end{align}
By marginalizing out the state of the other nodes, we can write this as
\begin{align}
p_i(&t+1)=\mathbb{E}_{\xi_{-i}(t)|\xi_i(t)=1}\bigg[1-\delta \prod_{j\in N_i}(1-\beta\mathbbm{1}_{\xi_j(t)=1})\bigg]p_i(t)\notag\\
&+\mathbb{E}_{\xi_{-i}(t)|\xi_i(t)=0}\bigg[1-\prod_{j\in N_i}(1-\beta\mathbbm{1}_{\xi_j(t)=1})\bigg](1-p_i(t)) ,
\end{align}
where the conditional expectations are on the joint probability of all nodes other than $i$ (denoted by $\xi_{-i}$).
\subsubsection{Approximated Nonlinear Model}
One may approximate $\prod_{j \in N_i} (1-\beta \mathbbm{1}_{\xi_j(t)=1})$ by averaging it as $\mathbb{E}[ 1-\beta \mathbbm{1}_{\xi_j(t)=1}] = 1 - \beta p_j(t)$ and using the assumption that the events are independent.
\begin{align}
P_i(t+1) & = \left(1-\delta \prod_{j \in N_i} \left( 1-\beta P_j(t) \right) \right)P_i(t) \nonumber \\
& \quad + \left(1- \prod_{j \in N_i} \left( 1-\beta P_j(t) \right) \right) (1-P_i(t)) \label{Eq:approximatedP}
\end{align} 
In fact this is the so-called mean-field approximation. We use capital $P$ for the approximated probabilities, to distinguish them from the exact probabilities of the Markov chain, $p$.

The approximated model is studied on $[0,1]^n$, the $n$-dimensional probability space which is computationally less demanding than the $2^n$-dimensional discrete space. One such model was studied by Chakrabarti and Wang \cite{chakrabarti2008epidemic}, \cite{WangEpidemic}. Ahn \cite{ahn2013global} viewed the $n$-dimensional probability distribution at time $t+1$ as the image of the probability distribution at time $t$ mapped by $\Phi:[0,1]^n \to [0,1]^n$. The $i$-th component of the epidemic map $\Phi$ is defined as follows:
\begin{empheq}[box=\fbox]{equation}
\Phi_i(x) = (1-\delta)x_i + ( 1- (1-\delta)x_i)\left( 1 - \prod_{j \in N_i} (1-\beta x_j) \right) \label{Eq:wtphieq}
\end{empheq}
It is trivial to check that $P_i(t+1) = \Phi_i((P_1(t), \cdots, P_n(t))^T)$ from \eqref{Eq:approximatedP}.

\subsubsection{Linear Model}
The linearization of the above nonlinear mapping around the origin is what is referred to as the linear model:
\begin{equation}
\tilde{P}_i(t+1) = (1-\delta)\tilde{P}_i(t) + \beta \left( \sum_{j \in N_i} \tilde{P}_j(t) \right) \label{Eq:linear}
\end{equation}
Putting together equations of this form for all $i$, one can see this as
\begin{empheq}[box=\fbox]{equation}
\tilde{P}(t+1) = ((1-\delta)I_n + \beta A) \tilde{P}(t)
\end{empheq}

Note that $(1-\delta)I_n + \beta A$ is in fact the Jacobian of $\Phi$ at the origin.

\subsection{Analysis of the Nonlinear Model}

\subsubsection{Epidemic Extinction: $\frac{\beta\lambda_{\max}(A)}{\delta}<1$}
We study the epidemic map $P_i(t+1) = \Phi_i((P_1(t), \cdots, P_n(t))^T)$ where $\Phi : [0,1]^n \to [0,1]^n$ is defined as \eqref{Eq:wtphieq} on $n$-dimensional probability space. To understand the behavior of this model, we can upper bound it as the following.
\begin{align}
\Phi_i(x) &= (1-\delta)x_i + ( 1- (1-\delta)x_i)\left( 1 - \prod_{j \in N_i} (1-\beta x_j) \right) \\
&\leq (1-\delta)x_i + \left( 1 - \prod_{j \in N_i} (1-\beta x_j) \right) \\
&\leq (1-\delta)x_i + \beta \left( \sum_{j \in N_i} x_j \right)
\end{align}
The latter equation is the linear map \eqref{Eq:linear}. In fact the linearization gives an upper bound on the nonlinear model.

For two real-valued column vectors $u,v\in \mathbb{R}^n$, we say $u \preceq v$, if $u_i \leq v_i$ for all $i \in \{1,\cdots,n\}$, and $u \prec v$, if $u_i < v_i$ for all $i \in \{1,\cdots,n\}$.
For $P(t)=(P_1(t), \cdots, P_n(t))^\mathrm{T}$
\begin{equation}
P(t+1) = \Phi(P(t)) \preceq ((1-\delta)I_n + \beta A ) P(t)
\end{equation}

Clearly $P(t)$ converges to the origin for both \eqref{Eq:wtphieq} and \eqref{Eq:linear} if $\lambda_{max}((1-\delta)I_n + \beta A ) <1 $. In other words, when $\frac{\beta\lambda_{\max}(A)}{\delta}$ is less than $1$, the origin is a unique fixed point of \eqref{Eq:wtphieq} which is globally stable. The reason is that this happens for the linear upper bound which is the Jacobian matrix of \eqref{Eq:wtphieq} at the origin. We will therefore focus on the dynamics of the system when  
$\lambda_{max}((1-\delta)I_n + \beta A ) > 1 $.
\subsubsection{Epidemic Spread: $\frac{\beta\lambda_{\max}(A)}{\delta}>1$}

\paragraph{Existence and Uniqueness of Nontrivial Fixed Point}
The origin, the trivial fixed point of the system equation is unstable when $\lambda_{max}((1-\delta)I_n + \beta A ) > 1 $.
Moreover, it is not clear in general whether there exists any other fixed point, or how many fixed points exist if so.
In this section, we prove that there actually exists a nontrivial fixed point of \eqref{Eq:seq3}.
We also prove that the nontrivial fixed point is unique.

Wang et al. \cite{WangEpidemic} and Chakrabarti et al. \cite{chakrabarti2008epidemic} focus on staying healthy by defining the probability that a node receives no infection from its neighborhood. We focus on {\em infection} rather than staying healthy.

Let $\Xi : [0,1]^n \to [0,1]^n$ with $\Xi=(\Xi_1 , \cdots, \Xi_n )^\mathrm{T}$ be a map associated with network $G$ satisfying the three properties below.

(a) $\Xi_i(x)=0$ and $\displaystyle \frac{ \partial \Xi_i }{ \partial x_j} = \beta A_{i,j}$ at the origin. 

(b) $\displaystyle \frac{ \partial \Xi_i }{ \partial x_j} > 0$ if $i \in N_j$ in $G$, and $\displaystyle \frac{ \partial \Xi_i}{ \partial x_j} = 0$ if $i \notin N_j$ in $G$. 

(c) For any $i,j,k \in \{ 1, \cdots, n \}$,  $\displaystyle \frac{ \partial^2 \Xi_i }{ \partial x_j \partial x_k} \leq 0$.

Obviously $\Xi_i(x)=\left( 1 - \prod_{j \in N_i} (1-\beta x_j) \right)$ satisfies all the conditions above. We define another map here. Let $\omega : [0,1] \to \mathbb{R}_+$ be a function which also satisfies three properties below.

(d) $\omega(0)=0$, $\omega(1) \geq 1$ 

(e) $\omega'(0)=\delta$, $\omega'(s) > 0$ for all $s \in (0,1)$ 

(f) $\displaystyle \frac{\omega(s_1)}{s_1} < \frac{\omega(s_2)}{s_2}$ if $s_1 < s_2$ 

It is also clear that $\displaystyle \omega(s) = \frac{\delta s}{1-(1-\delta)s}$ satisfies all three conditions above. By defining $\Xi(\cdot)$ and $\omega(\cdot)$ here, the analysis can also be applied directly to the immune-admitting model which will be described later.

We can view \eqref{Eq:wtphieq} as
\begin{equation}
P_i(t+1) = P_i(t) + (1-(1-\delta)P_i(t)) ( \Xi_i(P(t)) - \omega(P_i(t))) \label{Eq:seq3}
\end{equation}

\begin{lemma}\label{Lem:ccv}
Let $h_{i,u,v} : s \to \Xi_i(u+sv)$ be a function defined on subset of nonnegative real numbers 
for given $i \in \{1, \cdots, n\}$, $u,v \in [0,1]^n$. 
Then $\displaystyle \frac{h_{i,u,v}(s) - h_{i,u,v}(0)}{s}$ is a decreasing function of $s$.
\end{lemma}

\begin{lemma}\label{lm:v}
$\lambda_{max}((1-\delta)I_n + \beta A ) > 1 $ if and only if there exists $v \succ (0,\cdots, 0)^\mathrm{T} = 0_n$ such that $(\beta A - \delta I_n )v \succ 0_n$
\end{lemma}

The main theorem of this section which guarantees the existence and uniqueness of nontrivial fixed point of \eqref{Eq:seq3} follows.

\begin{theorem}\label{Thm:existence}
Define a map $\Psi : [0,1]^n \to \mathbb{R}^n$ with $\Xi$ and $\omega$ satisfying the conditions (a)-(f) above, as
\begin{equation}
\Psi_i(x) = \Xi_i(x) - \omega(x_i)~. \label{Eq:psieq}
\end{equation}
Then $\Psi=(\Psi_1, \cdots, \Psi_n)$ has a unique nontrivial (other than the origin) zero if $\frac{\beta\lambda_{\max}(A)}{\delta} > 1 $.
\end{theorem}

We should emphasize that this unique nontrivial zero (denoted by $x^*$ in the proof) is also the unique nontrivial fixed point of \eqref{Eq:seq3} as desired.

As a further remark, consider a network whose edge $\{i,j\}$ has weight $w_{ij}=w_{ji} \in [0,1]$. 
The weight of each edge could represent the degree of intimacy. 
The weight matrix can replace the adjacency matrix to define $\Xi_i(x)=\left( 1 - \prod_{j \in N_i} (1-\beta w_{ij} x_j) \right)$.
Then $\Xi$ defined by the weight matrix rather than the adjacency matrix also satisfies all three conditions (a)-(c) if $A_{ij}$ is replaced by $w_{ij}$ from (a). The system of equations will still have the same properties even if we admit different weights.

\paragraph{Global Stability of Nontrivial Fixed Point}
The origin, the trivial fixed point of the system is globally stable if $\lambda_{max}((1-\delta)I_n + \beta A ) < 1 $. The next issue is whether the nontrivial fixed point is also stable if $\lambda_{max}((1-\delta)I_n + \beta A ) > 1 $. It turns out that this is true, if we are not initially at the origin.

\begin{theorem}\label{Thm:GSofNFP}
Suppose $\lambda_{max}((1-\delta)I_n + \beta A ) > 1 $. 
As $t$ increases $P(t+1)=\Phi(P(t))$ defined by \eqref{Eq:wtphieq} converges to the unique nontrivial fixed point $x^*$, if $P(0)$ is not the origin.
\end{theorem}

\subsection{Analysis of the Exact Markov Chain}\label{sec:SIS_MC}
Returning back to the Markov chain model, we study the mixing time of the Markov chain and how it relates to the nonlinear and linear models. The mixing time of a Markov chain is defined as follows (\cite[Def.~4.5]{levin2009markov}):
\begin{equation}
t_{mix}(\epsilon)=\min \{t: \sup_\mu \| \mu S^t - \pi \|_{TV} \leq \epsilon \} ,  \label{Eq:mixingdefn}
\end{equation}
where $\mu$ is any initial probability distribution defined on the state space and $\pi$ is the stationary distribution. $\| \cdot \|_{TV}$ is total variation distance which measures distance of two probability distributions. Total variation distance of two probability measures $\mu$ and $\mu'$ is defined by
\begin{equation}
\| \mu -\mu' \|_{TV} = \frac{1}{2} \sum_x | \mu(x) - \mu'(x) |
\end{equation}
where $x$ is any possible state in the probability space. In fact $t_{mix}(\epsilon)$ is the smallest time where distance between the stationary distribution and probability distribution at time $t$ from any initial distribution is smaller than or equal to $\epsilon$. Roughly speaking, the mixing time measures how fast initial distribution converges to the limit distribution.

\subsubsection{A Linear Programming Approach}

Let $\mu(t) \in \mathbb{R}^{2^n}$ be a probability row vector of $\{0,1\}^n$ at time $t$. The probability that node $i$ is infected at time $t$, which is denoted by $p_i(t)$ as before, is simply the marginal probability of $\mu(t)$. That is $p_i(t)=\sum_{X_i=1} \mu_X(t)$. By defining $p_0(t)=1$ (for $\sum \mu_X(t)=1$) and sticking it to the rest of marginal probabilities, we get the column vector $p(t)=(p_0(t), p_1(t), \cdots, p_n(t) )^T$. One can interpret $p(t)$ as \emph{observable data} and $\mu(t)$ as \emph{hidden complete data} at time $t$. We give an upper bound for $p(t+1)$, observable data at the next time step, using only current observable information.

Let $f_i \in \mathbb{R}^{n+1}$ be the $i$-th unit column vector. $S$ is the transition matrix of the Markov chain, as defined before. $B \in \mathbb{R}^{2^n \times (n+1)}$ is a matrix that relates the observable data, $p(t)$, to the hidden complete data, $\mu(t)$. It can be formally expressed as:
\begin{equation}
B_{X,k} = \left\{
\begin{array}{rl}
1 & \text{if } k=0,\\
X_k & \text{if } k \in \{1,2, \cdots, n\}.
\end{array} \right.
\end{equation}

We would like to maximize $p_i(t+1)$ for a node $i$, given $p_1(t), \cdots, p_n(t)$. This leads to the following result.

\begin{proposition}\label{Lem:LP}
$\displaystyle p_i(t+1) \leq (1-\delta)p_i(t) + \beta \sum_{j \in N_i} p_j(t)$.
This is the tightest upper-bound that involves only the marginal probabilities at time $t$.
\end{proposition}

Notice that this is interestingly the linear model that we have been considering. In fact, by applying Proposition~\ref{Lem:LP} to each node, we can express it as
\begin{equation}
p(t+1) \preceq ((1-\delta)I_n + \beta A) p(t),
\end{equation}
and $(1-\delta)I_n + \beta A$ is the system matrix of the linear model.
For obtaining tighter bounds, one should use higher order terms than just marginals (e.g. pairwise probabilities, triples, etc.) \cite{ruhi2016improved}.

Now we prove the practical result of logarithmic mixing time for $\lambda_{max}((1-\delta)I_n + \beta A ) < 1 $. Let $e_X \in \mathbb{R}^{2^n}$ denote the $X$-th unit vector, i.e. the probability vector all of whose components are zero, except the $X$-th component. Also define $\bar{0}, \bar{1} \in \{0,1\}^n$ as the state where everyone is healthy and infected, respectively.
\begin{theorem}\label{Thm:upperboundmt}
If $\frac{\beta\lambda_{\max}(A)}{\delta}<1$, the mixing time of the Markov chain whose transition matrix $S$ is described by Eqs. \eqref{MC0} and \eqref{Eq:noimmune} is $O(\log n)$.
\end{theorem}

\subsubsection{Partial Ordering}\label{sec:partialordering}
In this section, we define a partial order on the set of probability vectors of $\{0,1\}^n$, and establish the connection between the nonlinear model and the Markov chain. The nonlinear model does not generally provide an upper bound on the true probabilities $p_i(t)$. However, it gives an upper bound on the probability that the system is not in the all-healthy state.

We define $\leq_{st}$ on the set of probability vectors of $\{0,1\}^n$ as follows. 
\begin{equation}
\mu \leq_{st} \mu' \quad \text{iff} \quad \sum_{X \preceq Z} \mu_X \geq \sum_{X \preceq Z} \mu'_X \quad \forall Z \in \{0,1\}^n
\end{equation}
where $X \preceq Z$ means $X_i \leq Z_i$ for all $i$. Note that $\sum_{X \preceq Z} \mu_X$ represents the probability that each node of $\mathbb{S}(Z)^c$ is healthy under probability distribution $\mu$. $\mu \leq_{st} \mu'$ means that the probability of some nodes being healthy is higher under $\mu$ than under $\mu'$, for any set of nodes. Roughly speaking, infection probability under $\mu'$ stochastically dominates one under $\mu$. It is trivial to check that $\leq_{st}$ is a well-defined partial order.
It is clear that $e_{\bar{1}}$ is the greatest element and $e_{\bar{0}}$ is the smallest element under $\leq_{st}$. As mentioned before, since the underlying graph is connected, and there we have an absorbing state, it is not hard to see that the stationary distribution is $e_{\bar{0}}$, which corresponds to all nodes being healthy with probability $1$. If all the nodes in the network are healthy, there is no infection and they always stay healthy.

The following two lemmas reveal why $\leq_{st}$ is nice; it makes $S$ an order-preserving map, i.e. $\mu \leq_{st} \mu'$ implies $\mu S \leq_{st} \mu' S$.

\begin{lemma}\label{Lem:RinverseSR}
$R^{-1}SR$ is a $2^n$ by $2^n$ matrix all of whose entries are non-negative where $R \in \mathbb{R}^{\{0,1\}^n \times \{0,1\}^n }$ is defined as
\begin{equation}
R_{X,Y}= \left\{
\begin{array}{rl}
1 & \text{if } X \preceq Y,\\
0 & \text{otherwise } 
\end{array} \right.
\end{equation}
\end{lemma}

\begin{lemma}\label{Lem:order}
If $\mu \leq_{st} \mu'$, then $\mu S \leq_{st} \mu' S$.
\end{lemma}
Note that Lemma~\ref{Lem:order} directly implies
$$\sum_{X \preceq \bar{0}} (\mu S^t)_X = (\mu S^t)_{\bar{0}} \geq (e_{\bar{1}} S^t)_{\bar{0}} = \sum_{X \preceq \bar{0}} (e_{\bar{1}} S^t)_X$$
for any probability vector $\mu$, since $\mu \leq_{st} e_{\bar{1}}$.

Now we establish a result which enables us to relate the nonlinear map $\Phi$ to the true probabilities of the Markov chain. For any given $n$-dimensional vector $r=(r_1, \cdots, r_n )^T$, define the $2^n$-dimensional column vector $u(r)$  by $u(r)_X = \displaystyle \prod_{i \in \mathbb{S}(X)} (1-r_i)$. Then we have the following lemma.
\begin{lemma}\label{Lem:uofr}
$Su(r) \succeq u(\Phi(r))$ for all $r \in [0,1]^n$.
\end{lemma}

It should be clear that $e_{\bar{0}}^T = u((1,1,\cdots,1)^T) = u(1_n)$ (we distinguish $1_n=(1,1,\cdots, 1)^T \in [0,1]^n$ from $\bar{1} \in \{0,1\}^n$ which is a state of infection). Lemma~\ref{Lem:uofr} is particularly useful because $S$ is a matrix all of whose entries are non-negative, and it follows that
\begin{equation}
S^t e_{\bar{0}}^T = S^t u(1_n) \succeq u( \Phi^t (1_n)) . \label{Eq:boundbyphitilde}
\end{equation}

Of note, by some algebra on $e_{\bar{1}} S^t e_{\bar{0}}^T$ using this bound, the same bound as in \eqref{epsilon_last} can be established, which leads to the mixing time result.

Furthermore, the $i$-th component of $\Phi^t(1_n)$ provides an upper bound on the probability that the current state is not the steady state, given that the infection started from node $i$ with probability 1 at time 0. Mathematically, $e_{\hat{i}}S^t e_{\bar{0}}^T \geq e_{\hat{i}} u(\Phi^t (1_n)) = 1- \Phi^t_i (1_n)$ by \eqref{Eq:boundbyphitilde}, and we have
\begin{align}
\mathbb{P}(\xi(t) \neq \bar{0}|\xi(0)=\hat{i}) &= 1- \mathbb{P}(\xi(t)=\bar{0}|\xi(0)=\hat{i})\\
&= 1- e_{\hat{i}}S^t e_{\bar{0}}^T\\
&\leq \Phi^t_i (1_n)
\end{align}

More importantly, the probability that the network is not in the all-healthy state at time $t$ given that the initial epidemic state is $X$ can be bounded above by the entries of $\Phi^t (1_n)$:
\begin{align}
&\mathbb{P}(\xi(t) \neq \bar{0}|\xi(0)=X) \\
&= 1- \mathbb{P}(\xi(t) = \bar{0}|\xi(0)=X) = 1- e_{X}S^t e_{\bar{0}}^T \\
&\leq 1- u( \Phi^t (1_n))_X= 1- \prod_{i \in \mathbb{S}(X)} \left( 1- \Phi^t_i(1_n) \right) \label{not-all-healthy}
\end{align}
\begin{proposition}
The nonlinear model provides an upper bound on the probability of the chain not being in the all-healthy state as
\begin{equation}
\mathbb{P}(\xi(t) \neq \bar{0}|\xi(0)=X)\leq 1- \prod_{i \in \mathbb{S}(X)} \left( 1- \Phi^t_i(1_n) \right)
\end{equation}
for any state $X$.
\end{proposition}

We should finally remark that the reason why it is possible for the nonlinear map to converge to a unique non-origin fixed point when $\frac{\beta\lambda_{\max}(A)}{\delta} > 1$, even though the original Markov chain model always converges to the all-healthy state, is that \eqref{not-all-healthy} is only an upper bound on $\mathbb{P}(\xi(t) \neq \bar{0}|\xi(0)=X)$.
In other words, if the origin is globally stable in the epidemic map $\Phi$, we can infer that the Markov chain model mixes fast. However, if the origin in the epidemic map is unstable, we cannot infer anything about mixing time.


\subsection{Generalized Contact Model}
In this section, we generalize the contact model. In the previous model, everyone had the same recovery rate $\delta$ and infection rate $\beta$. One of the main results was that the epidemic dies out fast if the largest eigenvalue of $M=(1-\delta)I_n + \beta A $ is smaller than $1$. $M$ is defined by $\beta$, the infection rate, $\delta$, the recovery rate, and $A$, the adjacency matrix. In other words $M$ is the contact model.

To model an epidemic spread where everyone has its own infection and recovery rate, we can define the generalized infection matrix. Let $M=(m_{i,j})$ be the generalized infection matrix where $m_{i,j} \in [0,1]$ represents the infection probability that $i$ is infected at time $t+1$ when $j$ is the only infected node at time $t$. In this setting, each diagonal entry $m_{i,i}$ represents self-infection rate. In other words, $1-m_{i,i}$ is recovery rate of node $i$ and $m_{i,i}$ is the probability that $i$ stays infected when there is no other infected nodes in the network. We also assume that probability of infection of each node given the current state $\xi(t)$ is independent. More precisely, for any two state vectors $X,Y \in \{0,1\}^n$,
\begin{equation}
\mathbb{P}(\xi(t+1)=Y|\xi(t)=X) = \prod_{i=1}^n \mathbb{P}(\xi_i(t+1)=Y_i|\xi(t)=X)
\end{equation}
Probability transition from given state is defined by $M$.
\begin{empheq}[box=\fbox]{align}
&\mathbb{P}(\xi_i(t+1)=Y_i|\xi(t)=X) \nonumber \\
&= \left\{
\begin{array}{rl} 
\displaystyle  \prod_{j \in \mathbb{S}(X)} (1-m_{i,j}) & \text{if } Y_i=0,\\
\displaystyle  1-\prod_{j \in \mathbb{S}(X)} (1-m_{i,j}) & \text{if } Y_i=1,\end{array} \right.
\end{empheq}
We define the transition matrix, $S^{(M)} \in \mathbb{R}^{\{0,1\}^n \times \{0,1\}^n}$ by $S^{(M)}_{X,Y} = \mathbb{P}(\xi_i(t+1)=Y_i|\xi(t)=X)$ in the equation above. For two probability distributions $\mu$ and $\mu'$ which are defined on $\{0,1\}^n$, $\mu \leq_{st} \mu'$ is equivalent to the statement that all the entries of $(\mu-\mu')R$ are non-negative. Lemma~\ref{Lem:RinverseSR} is also true for $S^{(M)}$. We can check that $(R^{-1}S^{(M)}R)_{X,Z} =  S^{(M^T)}_{\neg Z, \neg X} \geq 0$ where $M^T$ is the transpose of $M$. $S^{(M)}$ is an order-preserving map under $\leq_{st}$ by Lemma~\ref{Lem:order}. 

The epidemic map associated with $M$, $\displaystyle \Phi^{(M)}: [0,1]^n \to [0,1]^n$ is defined by
\begin{empheq}[box=\fbox]{equation}
\displaystyle \Phi^{(M)}_i(x)=1- \prod_{j=1}^n (1-m_{i,j}x_j)
\end{empheq}
and $\displaystyle \Phi^{(M)} = (\displaystyle \Phi^{(M)}_1 , \displaystyle \Phi^{(M)}_2, \cdots, \displaystyle \Phi^{(M)}_n)$. $M$ is the Jacobian matrix of $\displaystyle \Phi^{(M)}(\cdot)$ at the origin which gives an upper bound. i.e. $\Phi^{(M)}(x) \preceq Mx$. The origin is the unique fixed point which is globally stable if the largest eigenvalue of $M$ is smaller than $1$. It also has a unique nontrivial fixed point which is globally stable if the largest eigenvalue of $M$ is greater than $1$. 

Same as in Theorem~\ref{Thm:upperboundmt} and Lemma~\ref{Lem:uofr}, $\lambda_{\max}(M)<1$ guarantees that the mixing time of the Markov chain whose transition matrix is $S^{(M)}$ has an upper bound of $\displaystyle t_{mix}(\epsilon) \leq \frac{\log \frac{n}{\epsilon}}{-\log \| M \|}$, i.e. the mixing time is $O(\log n)$.

\subsection{Immune-Admitting Model}\label{sec:immune-admitting}
In this section, we study the immune-admitting model. The model is the same as that of the previous section except that in a single time interval a node cannot go from infected to healthy back to infected. In other words, a node does not get infected from its neighbors if it has just recovered from the disease. To summarize this,
\begin{empheq}[box=\mbox]{align}
\tikzmark{C}
&\mathbb{P}(\xi_i(t+1)=Y_i|\xi(t)=X) \nonumber \\
&= \left\{
\begin{array}{rl}
(1-\beta)^{m_i} & \text{if } (X_i,Y_i)=(0,0), \: |N_i \cap \mathbb{S}(X)| = {m_i},\\
1- (1-\beta)^{m_i} & \text{if } (X_i,Y_i)=(0,1), \: |N_i \cap \mathbb{S}(X)| = {m_i},\\
\delta & \text{if } (X_i,Y_i)=(1,0), \\
1-\delta & \text{if } (X_i,Y_i)=(1,1).
\end{array} \right.\tikzmark{D} \label{Eq:immuneadmitting}
\end{empheq}
\begin{tikzpicture}[remember picture,overlay]
\draw ([shift={(-.2em,2.5ex)}]C) rectangle ([shift={(-0.7em,-5.5ex)}]D);
\end{tikzpicture}

The transition matrix is defined in a similar way. In this model, the probability that a node becomes healthy from infected is $\delta$ which is larger than $\delta(1-\beta)^{m_i}$ as in immune-not-admitting model described in \eqref{Eq:noimmune}. Roughly speaking, the immune-admitting model is more likely to go to steady state than the immune-not-admitting model.

The mixing time of this model is also $O(\log n)$. Most of the formal proof is very similar to the one for immune-not-admitting model, and we omit it for the sake of brevity.


An epidemic map of the immune-admitting model can be studied as well, which is defined as 
\begin{empheq}[box=\fbox]{equation}
{\widetilde{\Phi}}_i(x) = (1-\delta)x_i + ( 1- x_i)\left( 1 - \prod_{j \in N_i} (1-\beta x_j) \right) \label{Eq:Phieq}
\end{empheq}

${\widetilde{\Phi}}:[0,1]^n \to [0,1]^n$ of \eqref{Eq:Phieq} has similar properties with $\Phi(\cdot)$ of \eqref{Eq:wtphieq}. ${\widetilde{\Phi}}(\cdot)$ and $\Phi(\cdot)$ have same Jacobian matrix at the origin which is linear upper bound of both nonlinear epidemic maps. Analysis of $\Phi(\cdot)$ is modified to analyze $\widetilde{\Phi}(\cdot)$ here. We represent $\widetilde{\Phi}(\cdot)$ using $\Xi(\cdot)$ and $\omega(\cdot)$ as we did in \eqref{Eq:seq3}. We can view 
\begin{equation}
{\widetilde{\Phi}}_i(x) = x_i + (1-x_i) ( \Xi_i(x) - \omega(x_i)) \label{Eq:Phiomegaform}
\end{equation}
where $\displaystyle \Xi_i(x)=\left( 1 - \prod_{j \in N_i} (1-\beta x_j) \right)$ and  $\displaystyle \omega(s) = \frac{\delta s}{1-s}$.
It is trivial to check that ${\Xi}(\cdot)$ and $\omega(\cdot)$ and satisfy all the conditions (a) - (f). Therefore we can apply Theorem~\ref{Thm:existence} to show that ${\widetilde{\Phi}}(\cdot)$ has a unique nontrivial fixed point if the largest eigenvalue of the Jacobian matrix at the origin is greater than $1$.

The origin, the trivial fixed point of the system is globally stable if $\lambda_{max}((1-\delta)I_n + \beta A ) < 1 $. The next issue is whether the unique nontrivial fixed point is also stable if $\lambda_{max}((1-\delta)I_n + \beta A ) > 1 $. This is not true in general for ${\widetilde{\Phi}}(\cdot)$. The following is an example of an unstable nontrivial fixed point.
\begin{equation}
\mathbf{A} = \left(
\begin{array}{ccc}
0 & 1 & 1 \\
1 & 0 & 0 \\
1 & 0 & 0
\end{array} \right) \qquad \delta=0.9 \quad \beta=0.9 \label{Eq:unstable}
\end{equation}

The nontrivial fixed point of the system above is $x^*= (0.286, 0.222, 0.222)^\mathrm{T}$. The Jacobian matrix of ${\widetilde{\Phi}}$ at $x^*$ is
\begin{equation} 
J_{{\widetilde{\Phi}}(x^*)} = \left(
\begin{array}{ccc}
-0.260 & 0.514 & 0.514 \\
0.700 & -0.157 & 0 \\
0.700 & 0 & -0.157
\end{array} \right)
\end{equation}

The eigenvalue with largest absolute value in the above Jacobian matrix is $-1.059$ whose absolute value is greater than 1. However, $P(t)={\widetilde{\Phi}}^t(P(0))$ converges to a cycle rather than a nontrivial fixed point $x^*$.

The biggest difference between \eqref{Eq:Phieq} and \eqref{Eq:wtphieq} is that $\displaystyle \frac{ \partial \Phi_i }{ \partial x_j} \geq 0$ for any $i,j \in \{1,\cdots, n\}$ in \eqref{Eq:wtphieq}, while it does not hold for ${\widetilde{\Phi}}(\cdot)$ in \eqref{Eq:Phieq}. The proof of Theorem~\ref{Thm:GSofNFP} can be applied to ${\widetilde{\Phi}}(\cdot)$ if $\displaystyle \frac{ \partial {\widetilde{\Phi}}_i }{ \partial x_j} \geq 0$ for any $i,j \in \{1,\cdots, n\}$ in \eqref{Eq:Phieq}.

Even though the nontrivial fixed point of ${\widetilde{\Phi}}(\cdot)$ is not stable generally, we shall show that it is stable with high probability for a family of random graphs. To study the stability of the nontrivial fixed point with high probability, we will begin with the following lemma that demonstrates that the Jacobian matrix at $x^*$ has no eigenvalue greater than or equal to unity for any values of $\beta$ and $\delta$ and for any connected graph.

\begin{lemma}\label{Lem:jacob}
Suppose that $x^*$ is a unique nontrivial fixed point of ${\widetilde{\Phi}} : [0,1]^n \to [0,1]^n$ with $\Xi$ satisfying the conditions (a),(b) and (c) when $\lambda_{max}((1-\delta)I_n + \beta A ) > 1 $. 
Then the Jacobian matrix of ${\widetilde{\Phi}}$ at $x^*$ has no eigenvalue of greater than or equal to $1$.
\end{lemma}
For the proof see pages 64--66 of \cite{ahn2014random}.

Even though $J_{\widetilde{\Phi}}$ has no eigenvalue which is greater than or equal to $1$, the fixed point $x^*$ still has a chance to be unstable if there is an eigenvalue which is greater than or equal to 1 in absolute value. We now show that $x^*$ is stable with high probability when we consider a certain family of random graphs and the number of vertices is large. We will later show that this family of random graphs includes Erd\"os-R\'enyi graphs.

We fix $\Xi_i(x)=\left( 1 - \prod_{j \in N_i} (1-\beta x_j) \right)$ from now on.
\begin{equation}
 \frac{ \partial \Xi_i }{ \partial x_j} = \beta \prod_{k \in N_i\setminus \{j\}} (1-\beta x_k) = \beta \frac{ 1 - \Xi_i}{1-\beta x_j}  \text{ if } i \in N_j\text{ in }G
\end{equation}
\begin{equation}
J_\Xi = \beta \,\text{diag}(1_n - \Xi) A \,\text{diag}(1_n - \beta x)^{-1}
\end{equation}

\begin{theorem}
Suppose that $G^{(n)}$ is a random graph with $n$ vertices and let 
$d^{(n)}_{\min}$ and $d^{(n)}_{\max}$ denote the minimum and maximum degree of $G^{(n)}$. 
If $\mathrm{Pr} [(d^{(n)}_{\min})^2 > a \cdot d^{(n)}_{\max} ]$ goes to $1$ as $n$ goes to infinity for any fixed $a > 0$, 
then the system is unstable at the origin and locally stable at the nontrivial fixed point $x^*$ with high probability as $n$ grows, for any fixed $\beta$ and $\delta$.
\label{lem:rg}
\end{theorem}
For the proof see pages 66--68 of \cite{ahn2014random}.

\begin{figure*}[thpb]
  \centering
    \includegraphics[width=1.4\columnwidth]{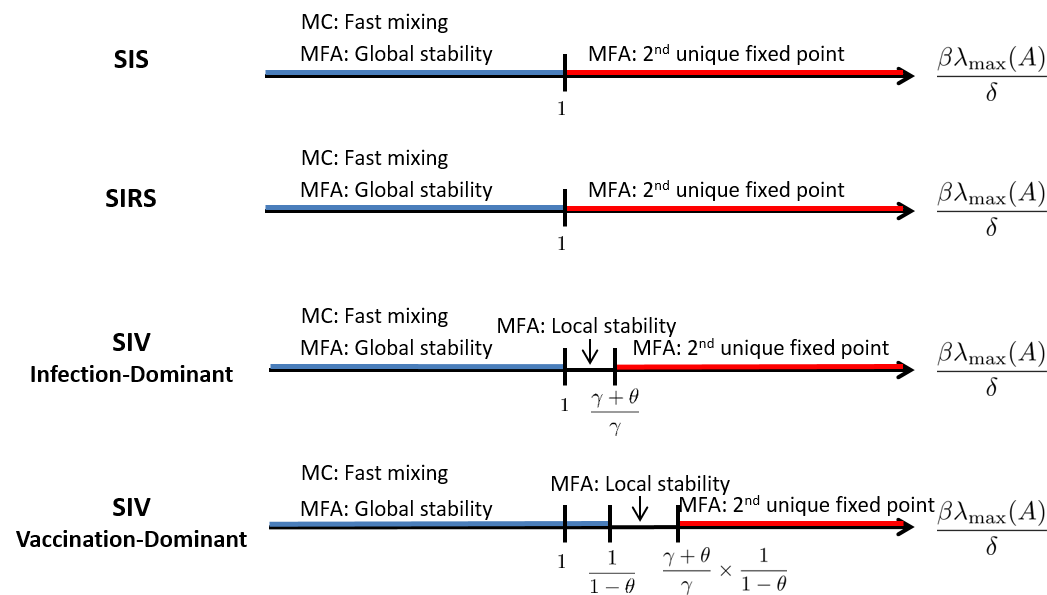}
      \caption{Summary and comparison of the results for SIS, SIRS, and SIV models, as a function of $\frac{\beta\lambda_{max}(A)}{\delta}$. MC stands for the Markov chain model. MFA stands for the mean-field approximation, aka the nonlinear model.}
      \label{fig:comp}
\end{figure*}

We can think of several random graph models that satisfy the condition of Theorem \ref{lem:rg}. For example, if the random graph has uniform degree then the minimum degree and maximum degree are identical and as long as the degree grows with $n$, the ratio $\displaystyle \frac{d_{\min}^2}{d_{\max}} = d$ will grow with any $n$ and exceed $a$ with high probability. Similarly,for random graphs where the degree distribution of each node is identical and the degree distribution "concentrates", so that 
we can expect that the maximum degree and the minimum degree are proportional to the expected of degree, in which case
$\displaystyle \frac{d_{\min}^2}{d_{\max}}$ grows if the expected degree increases unbounded with $n$.
The Erd\"os-R\'enyi random graph, $G^{(n)}=G(n,p(n))$ has identical degree distribution.
\begin{corollary}
Consider an Erd\"os-R\'enyi random graph $G^{(n)}=G(n,p(n))$ with $\displaystyle p(n)= c \frac{\log n}{n}$ where $c > 1$ is a constant. Then 
$\widetilde{\Phi}(\cdot)$ is locally unstable at the origin and has a locally stable nontrivial fixed point with high probability for any fixed $\beta$ and $\delta$.
\end{corollary}
For the proof see page 69 of \cite{ahn2014random}.

Since $\displaystyle p= c \frac{\log n}{n}$ for $c=1$ is also the threshold for connectivity,
we can say that connected Erd\"os-R\'enyi graphs have a nontrivial stable fixed point with high probability.

The random geometric graph $G^{(n)}=G(n,r(n))$ also has identical degree distribution if each node is distributed uniformly. As studied in \cite{PenroseRandomGeometricGraph}, such random graphs have maximum and minimum degree which are proportional to the expected degree with high probability if $r(n)$ is smaller than the threshold of connectivity. Like Erd\"os-R\'enyi graphs, it has high probability of having a nontrivial stable fixed point if the degree grows with $n$.

\section{SIRS Epidemics}
In this section we consider the SIRS model in which each node can be in one of three states of S, I and R. During each time epoch, nodes in the susceptible state can be infected by their infected neighbors according to independent events with probability $\beta$ (the {\em infection rate}) each. Nodes that are infected, during each such time epoch can recover with probability $\delta$ (the {\em recovery rate}) and, finally, nodes in the recovered state can randomly transition to the susceptible state with probability $\gamma$ ({\em immunization loss}).
\subsection{Model Description}
\subsubsection{Exact Markov Chain Model}
We start again with the exact Markov chain model. The state of node $i$ at time $t$, denoted by $\xi_i(t)$, can take one of the following values: $0$ for \emph{Susceptible} (or healthy), $1$ for \emph{Infected} (or Infectious), and $2$ for \emph{Recovered}. i.e. $\xi_i(t) \in \left\{0,1,2\right\}$. Fig. \ref{fig} shows the three states and the corresponding transitions. $\beta$ is the transmission probability on each link, $\delta$ is the healing probability, and $\gamma$ is the immunization loss probability.

\begin{figure}[htpb]
  \centering
    \includegraphics[width=0.7\columnwidth]{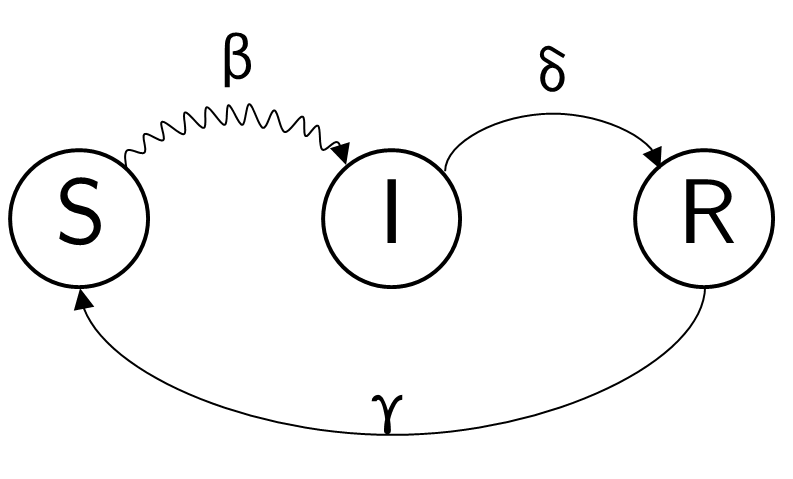}
      \caption{State diagram of a single node in the SIRS model, and the transition rates. Wavy arrow represents exogenous (network-based) transition. $S$ is healthy but can get infected, $I$ is infected, $R$ is healthy but cannot get infected.}
      \label{fig}
\end{figure}

The state of the whole network can be represented as:
\begin{equation}
\xi(t)=(\xi_i(t),\dots,\xi_n(t)) \in \left\{0,1,2\right\}^n
\end{equation}
Furthermore, let $S$ denote the $3^n\times 3^n$ state transition matrix of the Markov chain, with elements of the form:
\begin{align}\label{MC1}
S_{X,Y}&= \mathbb{P}\left(\xi(t+1)=Y \mid \xi(t) = X\right)\notag\\
 &= \prod_{i=1}^n \mathbb{P}\left(\xi_i(t+1)=Y_i \mid \xi(t) = X\right) ,
\end{align}
due to the independence of the next states given the current state.
\begin{empheq}[box=\fbox]{multline}\label{MC2}
\mathbb{P}\left(\xi_i(t+1)=Y_i \mid \xi(t) = X\right)=\\
\begin{cases}
(1-\beta)^{m_i},& \text{if } (X_i,Y_i)=(0,0)\\
1-(1-\beta)^{m_i},& \text{if } (X_i,Y_i)=(0,1)\\
0,& \text{if } (X_i,Y_i)=(0,2)\\
0,& \text{if } (X_i,Y_i)=(1,0)\\
1-\delta,& \text{if } (X_i,Y_i)=(1,1)\\
\delta,& \text{if } (X_i,Y_i)=(1,2)\\
\gamma,& \text{if } (X_i,Y_i)=(2,0)\\
0,& \text{if } (X_i,Y_i)=(2,1)\\
1-\gamma,& \text{if } (X_i,Y_i)=(2,2)\\
\end{cases} ,
\end{empheq}
where $m_i=\left\vert{\left\{ {j\in N_i} \mid X_j=1\right\}}\right\vert=\left\vert{N_i\cap I(t)}\right\vert$. The set of susceptible, infected, and recovered nodes at time $t$ are denoted as $S(t)$, $I(t)$, and $R(t)$ respectively.

We state the marginal probability of the nodes as $p_{R,i}(t)$ and $p_{I,i}(t)$, for the probability that \emph{node $i$ is in state $R$ at time $t$} and the probability that \emph{node $i$ is in state $I$ at time $t$}, respectively. Then $p_{S,i}(t)$ follows immediately as $1-p_{R,i}(t)-p_{I,i}(t)$. Based on the above-mentioned transition rates, we can calculate the two marginal probabilities as:
\begin{align}
p&_{R,i}(t+1) =(1-\gamma)p_{R,i}(t)+\delta p_{I,i}(t) ,\label{exact_R}\\
p&_{I,i}(t+1) =(1-\delta)p_{I,i}(t)\notag\\
&+\mathbb{E}_{ |\xi_i(t)=0}\bigg[1-\prod_{j\in N_i}(1-\beta\mathbbm{1}_{\xi_j(t)=1})\bigg](1-p_{R,i}(t)-p_{I,i}(t)) ,\label{exact_I}
\end{align}
As mentioned, the recursion for $p_{S,i}(t+1)$ can be found from $p_{S,i}(t)+p_{I,i}(t)+p_{R,i}(t)=1$.
\subsubsection{Nonlinear Model}
One may consider the mean-field approximation of the above marginal probabilities, which can be expressed as:
\begin{empheq}[box=\fbox]{align}
P&_{R,i}(t+1) =(1-\gamma)P_{R,i}(t)+\delta P_{I,i}(t) ,\label{nonlinear_R}\\
P&_{I,i}(t+1) =(1-\delta)P_{I,i}(t)+\notag\\
&\Big(1-\prod_{j\in N_i} (1-\beta P_{I,j}(t))\Big)(1-P_{R,i}(t)-P_{I,i}(t)) ,\label{nonlinear_I}
\end{empheq}
This approximate model is in fact a nonlinear mapping with $2n$ states (rather than $3^n$ states).

\subsubsection{Linear Model}
One step further would be to approximate the preceding equations by a linear model. Linearizing Eqs. \eqref{nonlinear_R} and \eqref{nonlinear_I} around the origin results in the following mapping:
\begin{align}
\tilde{P}_{R,i}(t+1) &=(1-\gamma)\tilde{P}_{R,i}(t)+\delta \tilde{P}_{I,i}(t) ,\\
\tilde{P}_{I,i}(t+1) &=(1-\delta)\tilde{P}_{I,i}(t)+\beta\sum\limits_{j\in N_i} \tilde{P}_{I,j} .
\end{align}
These equations (for all $i$) can be expressed in a matrix form:
\begin{empheq}[box=\fbox]{align}\label{linear}
&\begin{bmatrix}\tilde{P}_R(t+1)\\\tilde{P}_I(t+1)\end{bmatrix}=M \begin{bmatrix}\tilde{P}_R(t)\\\tilde{P}_I(t)\end{bmatrix} ,\\
&\text{where}\notag\\
&M = \begin{bmatrix}
(1-\gamma)I_n & \delta I_n\\
0_{n\times n} & (1-\delta)I_n+\beta A
\end{bmatrix} .
\end{empheq}

\subsection{Analysis of the Nonlinear Model}
\subsubsection{Epidemic Extinction: $\frac{\beta\lambda_{\max}(A)}{\delta}<1$}
The origin is trivially a fixed point of both the linear (Eq. \ref{linear}) and nonlinear (Eqs. \ref{nonlinear_R} and \ref{nonlinear_I}) mappings. In fact, at this fixed point we have:
$$[P_{R,1}(t),  \dots, P_{R,n}(t), P_{I,1}(t), \dots, P_{I,n}(t)]^T= 0_{2n} ,$$
which means all the nodes are susceptible (healthy) with probability 1, and the system stays there permanently, because there are no infected nodes anymore.

Clearly, if $\|M\|<1$, then the origin is globally stable for the linear model (\ref{linear}) and also locally stable for the nonlinear model (\ref{nonlinear_I}, \ref{nonlinear_R}). The eigenvalues of $M$ matrix consist of the eigenvalues of $(1-\gamma)I_n$ and the eigenvalues of $(1-\delta)I_n+\beta A$. Noticing that the eigenvalues of $(1-\gamma)I_n$ are always less than one, it can be concluded that $\|M\|<1$ if the largest eigenvalue of $(1-\delta)I_n+\beta A$ is less than one.

In addition, the linear model (\ref{linear}) is an upper bound on the nonlinear model (\ref{nonlinear_R}, \ref{nonlinear_I}), i.e.
\begin{multline}
P_{I,i}(t+1) =(1-\delta)P_{I,i}(t)\\
+\Big(1-\prod_{j\in N_i} (1-\beta P_{I,j}(t))\Big)(1-P_{R,i}(t)-P_{I,i}(t))\\
\leq (1-\delta)P_{I,i}(t)+\beta\sum\limits_{j\in N_i} P_{I,j} ,
\end{multline}
This concludes the following.
\begin{proposition}
If $\frac{\beta\lambda_{\max}(A)}{\delta}<1$, then the origin is a globally stable fixed point for both linear model (\ref{linear}) and nonlinear model (\ref{nonlinear_R}, \ref{nonlinear_I}).
\end{proposition}

\subsubsection{Epidemic Spread: $\frac{\beta\lambda_{\max}(A)}{\delta}>1$}
\paragraph{Existence and Uniqueness of Nontrivial Fixed Point}\label{sec:SIRS_secondfixedpoint}
The trivial fixed point of the mappings, the origin, is not stable if $(1-\delta)+\beta\lambda_{\max}(A)>1$. We show that there exists a unique nontrivial fixed point when $(1-\delta)+\beta\lambda_{\max}(A)>1$ for SIRS model.

By rearranging Eq. \eqref{nonlinear_I}, we can rewrite the system equations as:
\begin{empheq}[left=\empheqlbrace]{align}
P_{R,i}(t+1)=&(1-\gamma)P_{R,i}(t)+\delta P_{I,i}(t)\label{nonlinear_R_new}\\
P_{I,i}(t+1)=&P_{I,i}(t)+(1-P_{R,i}(t)-P_{I,i}(t))\notag\\
&\cdot\big(\Xi_i(P_I(t))-\omega(P_{R,i}(t),P_{I,i}(t))\big) ,\label{nonlinear_I_new}
\end{empheq}
where $\Xi_i \colon [0,1]^n \to [0,1]$ and $\omega \colon [0,1]^2 \to \mathbb{R}^+$ are the following maps associated with network $G$:
\begin{equation}
\Xi_i(P_I(t))=1-\prod_{j\in N_i} (1-\beta P_{I,j}(t)) ,
\end{equation}
\begin{equation}
\omega(P_{R,i}(t),P_{I,i}(t))=\frac{\delta P_{I,i}(t)}{1-P_{R,i}(t)-P_{I,i}(t)} .
\end{equation}

It can be verified that the maps defined above, enjoy the following properties:
\begin{enumerate}[label=(\alph*)]
\item $\Xi_i(0_n)=0$\\ $\frac{\partial \Xi_i(P_I)}{\partial P_{I,j}}\bigg|_{0_n}=\beta A_{i,j}$
\item $\begin{cases}\frac{\partial \Xi_i(P_I)}{\partial P_{I,j}}>0 &\mbox{if } i\in N_j\\ \frac{\partial \Xi_i(P_I)}{\partial P_{I,j}}=0 &\mbox{if } i\not\in N_j\end{cases}$
\item $\frac{\partial^2\Xi_i(P_I)}{\partial P_{I,j} \partial P_{I,k}}\leq 0 \quad \forall i,j,k \in \{1,\dots,n\}$
\item $\omega(0,0)=0$\\ $\frac{\partial \omega(P_{R,i},P_{I,i})}{\partial P_{I,i}}\bigg|_{(0,0)}=\delta$
\item $\frac{\partial \omega(P_{R,i},P_{I,i})}{\partial P_{I,i}}>0 \quad \forall P_{I,i}\in(0,1)$
\item $\frac{\omega(P_{R,i},P_{I,i})}{P_{I,i}}$ is an increasing function of both $P_{R,i}$ and $P_{I,i}$. More specifically: $\frac{\omega(s_1,t_1)}{s_1}<\frac{\omega(s_2,t_2)}{s_2}$ if $s_1<s_2$ and $t_1<t_2$.
\end{enumerate}

The main result of this section is as follows.
\begin{theorem}\label{thm:SIRS_nontrivial}
If $\frac{\beta\lambda_{\max}(A)}{\delta}>1$, the nonlinear map (\ref{nonlinear_R}, \ref{nonlinear_I}), or equivalently (\ref{nonlinear_R_new}, \ref{nonlinear_I_new}), has a unique nontrivial fixed point.
\end{theorem}

\paragraph{Stability of the Nontrivial Fixed Point}
Since the trivial fixed point was globally stable when $\frac{\beta\lambda_{\max}(A)}{\delta}<1$, the existence of a second unique fixed point at $\frac{\beta\lambda_{\max}(A)}{\delta}>1$ raises the question of whether it is also stable. It turns out that this is not true in general. In fact, same as in immune-admitting SIS model (Section~\ref{sec:immune-admitting}), we can find simple examples in which the system converges to a cycle rather than the unique second fixed point.
Nevertheless, like immune-admitting SIS, this fixed point can be shown to be stable with high probability for some general families of random graphs.

\subsection{Analysis of the Exact Markov Chain}
Since the graph $G$ is connected and the Markov chain has an absorbing state $\xi=(0,0,\dots,0) = \bar{0}$, the unique stationary distribution is:
$$\pi=e_{\bar{0}} ,$$
where $e_X \in \mathbb{R}^{3^n}$ denotes the probability vector with all elements of zero, except the $X$-th one. This coincides with the fixed point of the mappings; however, the main concern is whether the Markov chain converges to its stationary distribution within a ``reasonable amount of time,'' or not.

We show that when $\frac{\beta\lambda_{\max}(A)}{\delta}<1$, not only are the linear and nonlinear maps globally stable at the origin, but also the mixing time of the Markov chain is $O(\log n )$, meaning that the Markov chain mixes fast and the epidemic dies out.

Let the row vector $\mu(t)\in \mathbb{R}^{3^n}$ be the probability vector of the Markov chain. The relationship between these probabilities ($\mu_X(t)$) and the marginal probabilities ($p_{R,i}(t)$, $p_{I,i}(t)$) is in the following forms: $p_{R,i}(t)=\sum_{X_i=2} \mu_X(t)$, $p_{I,i}(t)=\sum_{X_i=1} \mu_X(t)$. We express all these terms as well as $p_0=\sum \mu_X(t)=1$ in the form of a column vector $p(t)=[p_0(t),p_1(t),\dots,p_{2n}]^T$, i.e.
\begin{equation}
p(t)=\begin{bmatrix} 1, \vline p_{R,1}(t), \dots,p_{R,n}(t), \vline p_{I,1}(t), \dots, p_{I,n}(t)\end{bmatrix}^T .
\end{equation}
The matrix $B \in \mathbb{R}^{3^n\times (2n+1)}$ which relates the ``observable data'' $p(t)$, and the ``hidden complete data'' $\mu(t)$, can be expressed as:
\begin{equation}
B_{X,k}=
\begin{cases}
1,& \text{if } k=0\\ \hdashline
0,& \text{if } k\in\left\{1,2,\dots,n\right\} \text{ and } X_k=0\\
0,& \text{if } k\in\left\{1,2,\dots,n\right\} \text{ and } X_k=1\\
1,& \text{if } k\in\left\{1,2,\dots,n\right\} \text{ and } X_k=2\\ \hdashline
0,& \text{if } k\in\left\{n+1,n+2,\dots,2n\right\} \text{ and } X_{k-n}=0\\
1,& \text{if } k\in\left\{n+1,n+2,\dots,2n\right\} \text{ and } X_{k-n}=1\\
0,& \text{if } k\in\left\{n+1,n+2,\dots,2n\right\} \text{ and } X_{k-n}=2\\
\end{cases}
\end{equation}

Now we can proceed to the main theorem of this section.
\begin{theorem}\label{thm_mixing}
If $\frac{\beta\lambda_{\max}(A)}{\delta}<1$, the mixing time of the Markov chain whose transition matrix $S$ is described by Eqs. \eqref{MC1} and \eqref{MC2} is $O(\log n)$.
\end{theorem}

\section{SIV Epidemics}
In this section we consider the effect of vaccination by incorporating direct immunization into the model studied in the previous sections. In other words, the transition from $S$ to $R$ is also permitted now (See Fig. \ref{fig2}). This class of processes are often referred to as SIV (Susceptible-Infected-Vaccinated) epidemics. Depending on the value of $\gamma$, this model can represent temporary ($\gamma\neq 0$) or permanent ($\gamma=0$) immunization. Moreover, based on the efficacy of the vaccine, there are two different models: infection-dominant and vaccination-dominant.

\begin{figure}[thpb]
  \centering
    \includegraphics[width=0.7\columnwidth]{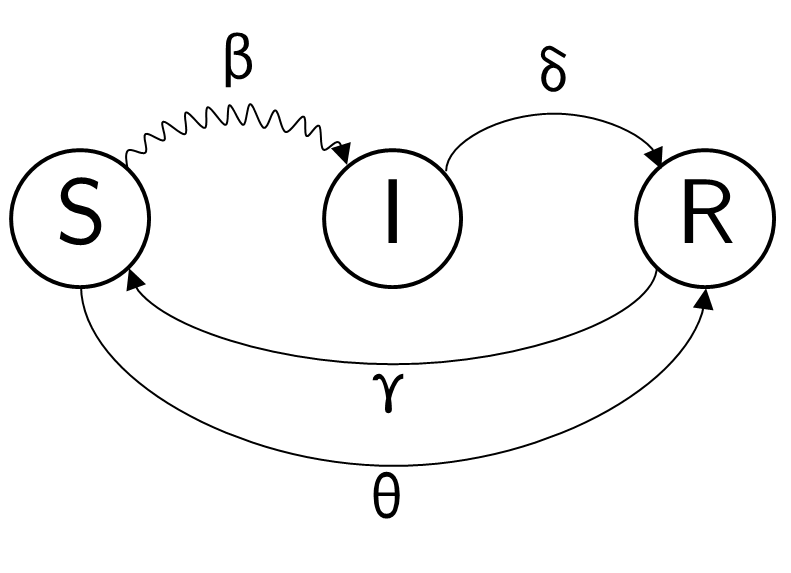}
      \caption{State diagram of a single node in the SIRS-with-Vaccination model, and the transition rates. Wavy arrow represents exogenous (network-based) transition. $\theta$ represents the probability of direct immunization.}
      \label{fig2}
\end{figure}

\subsection{Infection-Dominant Model}
In this case, the infection is dominant, in the sense that if a susceptible node receives both infection and vaccine at the same time, it gets infected. The elements of state transition matrix are
\begin{align}\label{MC1_id}
S_{X,Y}&= \mathbb{P}\left(\xi(t+1)=Y \mid \xi(t)= X\right)\notag\\
&= \prod_{i=1}^n \mathbb{P}\left(\xi_i(t+1)=Y_i \mid \xi(t) = X\right) ,
\end{align}
where
\begin{empheq}[box=\fbox]{multline}\label{MC2_id}
\mathbb{P}\left(\xi_i(t+1)=Y_i \mid \xi(t) = X\right)=\\
\begin{cases}
(1-\beta)^{m_i}(1-\theta),& \text{if } (X_i,Y_i)=(0,0)\\
1-(1-\beta)^{m_i},& \text{if } (X_i,Y_i)=(0,1)\\
(1-\beta)^{m_i}\theta,& \text{if } (X_i,Y_i)=(0,2)\\
0,& \text{if } (X_i,Y_i)=(1,0)\\
1-\delta,& \text{if } (X_i,Y_i)=(1,1)\\
\delta,& \text{if } (X_i,Y_i)=(1,2)\\
\gamma,& \text{if } (X_i,Y_i)=(2,0)\\
0,& \text{if } (X_i,Y_i)=(2,1)\\
1-\gamma,& \text{if } (X_i,Y_i)=(2,2)\\
\end{cases} ,
\end{empheq}
and as before $m_i=\left\vert{\left\{ {j\in N_i} \mid X_j=1\right\}}\right\vert=\left\vert{N_i\cap I(t)}\right\vert$. Compared to Eq. \eqref{MC2}, the first and the third elements have changed in Eq. \eqref{MC2_id}, and for $\theta=0$ the model reduces to the previous one.

In this infection-dominant model the marginal probabilities are:
\begin{align}
p&_{R,i}(t+1) =(1-\gamma)p_{R,i}(t)+\delta p_{I,i}(t)\notag\\
&+\mathbb{E}_{ |\xi_i(t)=0}\bigg[\prod_{j\in N_i}(1-\beta\mathbbm{1}_{\xi_j(t)=1})\bigg]\theta(1-p_{R,i}(t)-p_{I,i}(t)) ,\label{exact_R_id}\\
p&_{I,i}(t+1) =(1-\delta)p_{I,i}(t)\notag\\
&+\mathbb{E}_{ |\xi_i(t)=0}\bigg[1-\prod_{j\in N_i}(1-\beta\mathbbm{1}_{\xi_j(t)=1})\bigg](1-p_{R,i}(t)-p_{I,i}(t)) ,\label{exact_I_id}
\end{align}

The steady state behavior in the presence of immunization is rather different from the SIS/SIRS cases, in which all the nodes became susceptible. In this model, once there is no node in the infected state, the Markov chain reduces to a simpler Markov chain, where the nodes are all decoupled. In fact from that time on, each node has an independent transition probability between $S$ and $R$. The stationary distribution of each single node is then $P_S^* = \frac{\gamma}{\gamma+\theta}$ and $P_R^* = \frac{\theta}{\gamma+\theta}$ (Fig. \ref{steady}). In order for this MC to converge, we should have $\gamma\theta \neq 1$. The stationary distribution of each state $X$ is then:
$$\pi_X = \prod_{i=1}^n (\frac{\gamma}{\gamma+\theta})^{\mathbb{I}(X_i=0)} \cdot 0^{\mathbb{I}(X_i=1)} \cdot (\frac{\theta}{\gamma+\theta})^{\mathbb{I}(X_i=2)}$$
\begin{figure}[thpb]
  \centering
    \includegraphics[width=0.8\columnwidth]{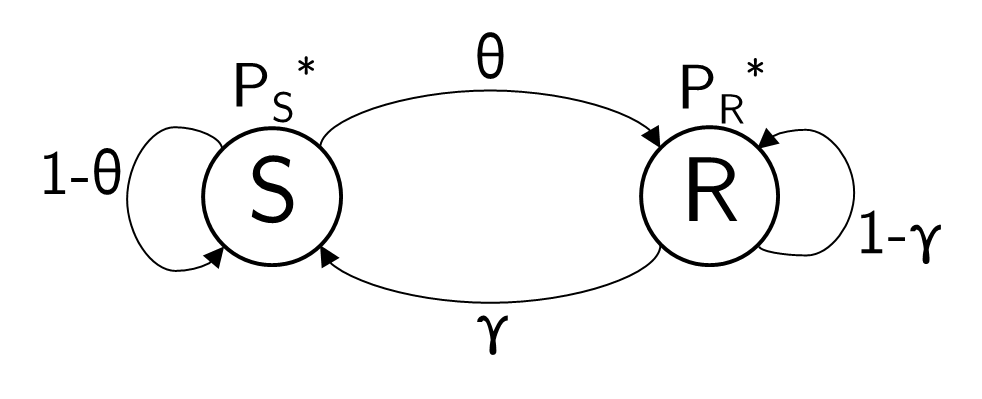}
      \caption{Reduced Markov chain of a single node in the steady state.}
      \label{steady}
\end{figure}

Now the nonlinear map (mean-field approximation of the Markov chain model) can be obtained as:
\begin{empheq}[box=\fbox]{align}
P&_{R,i}(t+1) =(1-\gamma)P_{R,i}(t)+\delta P_{I,i}(t)\notag\\
&+ \prod_{j\in N_i} (1-\beta P_{I,j}(t))\theta(1-P_{R,i}(t)-P_{I,i}(t)), \label{nonlinear_R_id}\\
P&_{I,i}(t+1) =(1-\delta)P_{I,i}(t)+\notag\\
&\Big(1-\prod_{j\in N_i} (1-\beta P_{I,j}(t))\Big)(1-P_{R,i}(t)-P_{I,i}(t)). \label{nonlinear_I_id}
\end{empheq}

It can be easily verified that one fixed point of this nonlinear map occurs at $P_{R,i}(t)=P_R^*$ and $P_{I,i}(t)=0$, i.e.
$$\begin{bmatrix}P_{R}(t)\\ P_{I}(t)\end{bmatrix}= \begin{bmatrix} \frac{\theta}{\gamma+\theta}1_n\\ 0_n\end{bmatrix}, $$
which is nicely consistent with the steady state of the Markov chain.

After some algebra, the linearization of the above model around the fixed point can be expressed as:
\begin{empheq}[box=\fbox]{align}
&\begin{bmatrix}\tilde{P}_R(t+1)\\\tilde{P}_I(t+1)\end{bmatrix}=\begin{bmatrix}  P_R^* 1_n\\ 0_n\end{bmatrix} + M \begin{bmatrix}\tilde{P}_R(t)-P_R^* 1_n\\\tilde{P}_I(t) - 0_n\end{bmatrix} ,\\
&\text{where}\notag\\
&M=\begin{bmatrix}
(1-\gamma-\theta)I_n & (\delta-\theta) I_n - \theta  P_S^* \beta A\\
0_{n\times n} & (1-\delta)I_n+ P_S^*\beta A
\end{bmatrix} .
\end{empheq}

\subsubsection{Analysis of the Nonlinear Model}

\paragraph{Epidemic Extinction: $\frac{\gamma}{\gamma+\theta}\frac{\beta\lambda_{\max}(A)}{\delta}<1$}

The following result summarizes the stability of the disease-free fixed point.
\begin{proposition}\label{prop_id}
The main fixed point of the nonlinear map (\ref{nonlinear_R_id}, \ref{nonlinear_I_id}) is
\begin{enumerate}[label=\alph*)]
\item locally stable, if $\frac{\gamma}{\gamma+\theta}\frac{\beta}{\delta}\lambda_{max}(A)<1$, and
\item globally stable, if $\frac{\beta}{\delta}\lambda_{max}(A)<1$ .
\end{enumerate}
\end{proposition}
The authors of \cite{prakash2012threshold} have shown the same condition for the local stability, but they do not provide any result on the global stability.

\paragraph{Epidemic Spread: $\frac{\gamma}{\gamma+\theta}\frac{\beta\lambda_{\max}(A)}{\delta}>1$}
The main fixed point of the mapping is not stable if $\frac{\gamma}{\gamma+\theta}\frac{\beta\lambda_{\max}(A)}{\delta}>1$. We show the existence and uniqueness of the the second fixed point for this case.

The gist of the proof is the same as that of Section \ref{sec:SIRS_secondfixedpoint}, except we replace Property (d) with the more general of:
\begin{enumerate}[label=(\alph*')]
\addtocounter{enumi}{3}
\item $\omega(P_{R,i},0)=0$\\ $\frac{\partial \omega(P_{R,i},P_{I,i})}{\partial P_{I,i}}\bigg|_{(P_{R,i},0)}=\frac{\delta}{1-P_{R,i}}$ 
\end{enumerate}
for any $P_{R,i}\neq 1$.

\begin{theorem}\label{thm:SIV_id_nontrivial}
If $\frac{\gamma}{\gamma+\theta}\frac{\beta\lambda_{\max}(A)}{\delta}>1$, the nonlinear map (\ref{nonlinear_R_id}, \ref{nonlinear_I_id}), has a unique nontrivial fixed point.
\end{theorem}

\subsubsection{Analysis of the Exact Markov Chain}
We show the mixing time result for this case as well. Vectors $\mu(t)$, $p(t)$ and the matrix $B$ are defined as before.

\begin{theorem}\label{thm_mixing_id}
If $\frac{\beta\lambda_{\max}(A)}{\delta}<1$, the mixing time of the Markov chain whose transition matrix $S$ is described by Eqs. \eqref{MC1_id} and \eqref{MC2_id} is $O(\log n)$.
\end{theorem}

\subsection{Vaccination-Dominant Model}
In this variation of the model the assumption is if a susceptible node receives both infection and vaccine at the same time, it becomes vaccinated. The transition probabilities of the Markov chain are again
\begin{align}\label{MC1_vd}
S_{X,Y}&= \mathbb{P}\left(\xi(t+1)=Y \mid \xi(t)= X\right)\notag\\
&= \prod_{i=1}^n \mathbb{P}\left(\xi_i(t+1)=Y_i \mid \xi(t) = X\right) ,
\end{align}
with the change that
\begin{empheq}[box=\fbox]{multline}\label{MC2_vd}
\mathbb{P}\left(\xi_i(t+1)=Y_i \mid \xi(t) = X\right)=\\
\begin{cases}
(1-\beta)^{m_i}(1-\theta),& \text{if } (X_i,Y_i)=(0,0)\\
(1-(1-\beta)^{m_i})(1-\theta),& \text{if } (X_i,Y_i)=(0,1)\\
\theta,& \text{if } (X_i,Y_i)=(0,2)\\
0,& \text{if } (X_i,Y_i)=(1,0)\\
1-\delta,& \text{if } (X_i,Y_i)=(1,1)\\
\delta,& \text{if } (X_i,Y_i)=(1,2)\\
\gamma,& \text{if } (X_i,Y_i)=(2,0)\\
0,& \text{if } (X_i,Y_i)=(2,1)\\
1-\gamma,& \text{if } (X_i,Y_i)=(2,2)\\
\end{cases} ,
\end{empheq}
and $m_i=\left\vert{\left\{ {j\in N_i} \mid X_j=1\right\}}\right\vert=\left\vert{N_i\cap I(t)}\right\vert$ as before.

In this case the marginal probabilities are:
\begin{align}
p_{R,i}&(t+1) =(1-\gamma)p_{R,i}(t)+\delta p_{I,i}(t)+\notag\\
& \theta(1-p_{R,i}(t)-p_{I,i}(t)) ,\label{exact_R_vd}\\
p_{I,i}&(t+1) =(1-\delta)p_{I,i}(t)+ (1-\theta)\times\notag\\
&\mathbb{E}_{ |\xi_i(t)=0}\bigg[1-\prod_{j\in N_i}(1-\beta\mathbbm{1}_{\xi_j(t)=1})\bigg](1-p_{R,i}(t)-p_{I,i}(t))\label{exact_I_vd}
\end{align}

The nonlinear map, or the mean-field approximation, can be stated as:
\begin{empheq}[box=\fbox]{align}
&P_{R,i}(t+1) =(1-\gamma)P_{R,i}(t)+\delta P_{I,i}(t)\notag\\
&+ \theta(1-P_{R,i}(t)-P_{I,i}(t)), \label{nonlinear_R_vd}\\
&P_{I,i}(t+1) =(1-\delta)P_{I,i}(t)+(1-\theta)\notag\\
&\cdot\Big(1-\prod_{j\in N_i} (1-\beta P_{I,j}(t))\Big)(1-P_{R,i}(t)-P_{I,i}(t)) \label{nonlinear_I_vd}
\end{empheq}

As a result, the first order (linear) model is:
\begin{empheq}[box=\fbox]{align}
&\begin{bmatrix}\tilde{P}_R(t+1)\\\tilde{P}_I(t+1)\end{bmatrix}=\begin{bmatrix} P_R^* 1_n\\ 0_n\end{bmatrix} + M \begin{bmatrix}\tilde{P}_R(t)-P_R^* 1_n\\\tilde{P}_I(t) - 0_n\end{bmatrix} ,\notag\\
&\text{where}\notag\\
&M=\begin{bmatrix}
(1-\gamma-\theta)I_n & (\delta-\theta) I_n - \theta  P_S^* \beta A\\
0_{n\times n} & (1-\delta)I_n+(1-\theta) P_S^*\beta A
\end{bmatrix} .\notag
\end{empheq}
We should note that for the vaccination-dominant model, the steady state of the Markov chain and the main fixed point of the mapping are exactly the same as in  the infection-dominant model. However, as we may expect, the vaccination-dominant model is more stable.
\subsubsection{Analysis of the Nonlinear Model}
\paragraph{Epidemic Extinction: $(1-\theta)\frac{\gamma}{\gamma+\theta}\frac{\beta}{\delta}\lambda_{max}(A)<1$}
The stability of the vaccination-dominant model can be summarized in the following theorem.
\begin{proposition}\label{prop_vd}
The main fixed point of the nonlinear map (\ref{nonlinear_R_vd}, \ref{nonlinear_I_vd}) is
\begin{enumerate}[label=\alph*)]
\item locally stable, if $(1-\theta)\frac{\gamma}{\gamma+\theta}\frac{\beta}{\delta}\lambda_{max}(A)<1$, and
\item globally stable, if $(1-\theta)\frac{\beta}{\delta}\lambda_{max}(A)<1$ .
\end{enumerate}
\end{proposition}

\begin{figure}[tpb]
  \centering
    \includegraphics[width=0.9\columnwidth]{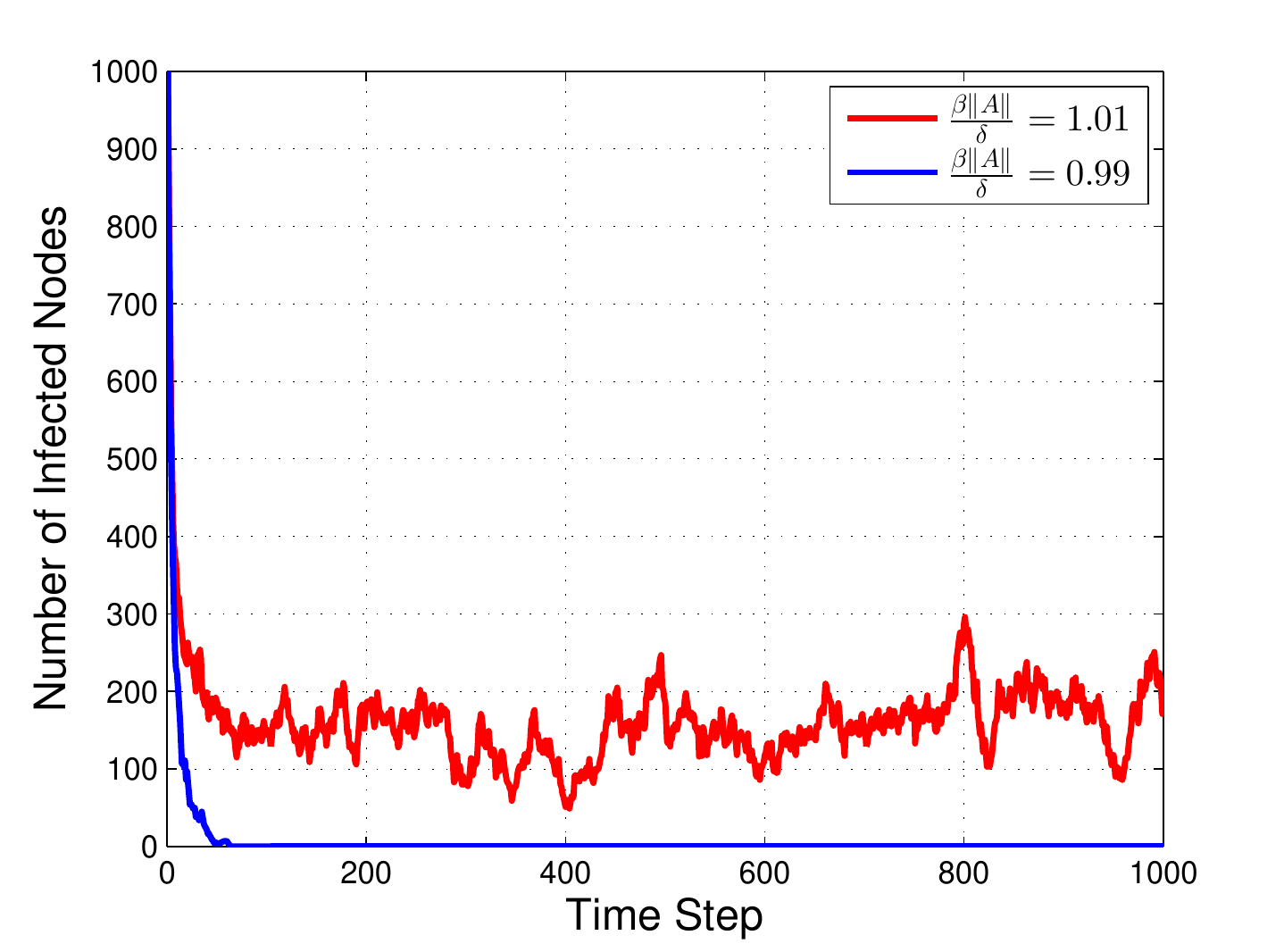}
      \caption{The evolution of an SIS epidemic over an Erd\H{o}s-R\'enyi graph with $n=2000$ nodes. Below the threshold we observe fast extinction of the epidemic (blue curve). Above the threshold, convergence is not observed (red curve).}
      \label{plot0}
\end{figure}

\begin{figure*}[!t]
\centering
\subfloat{\includegraphics[width=2.3in]{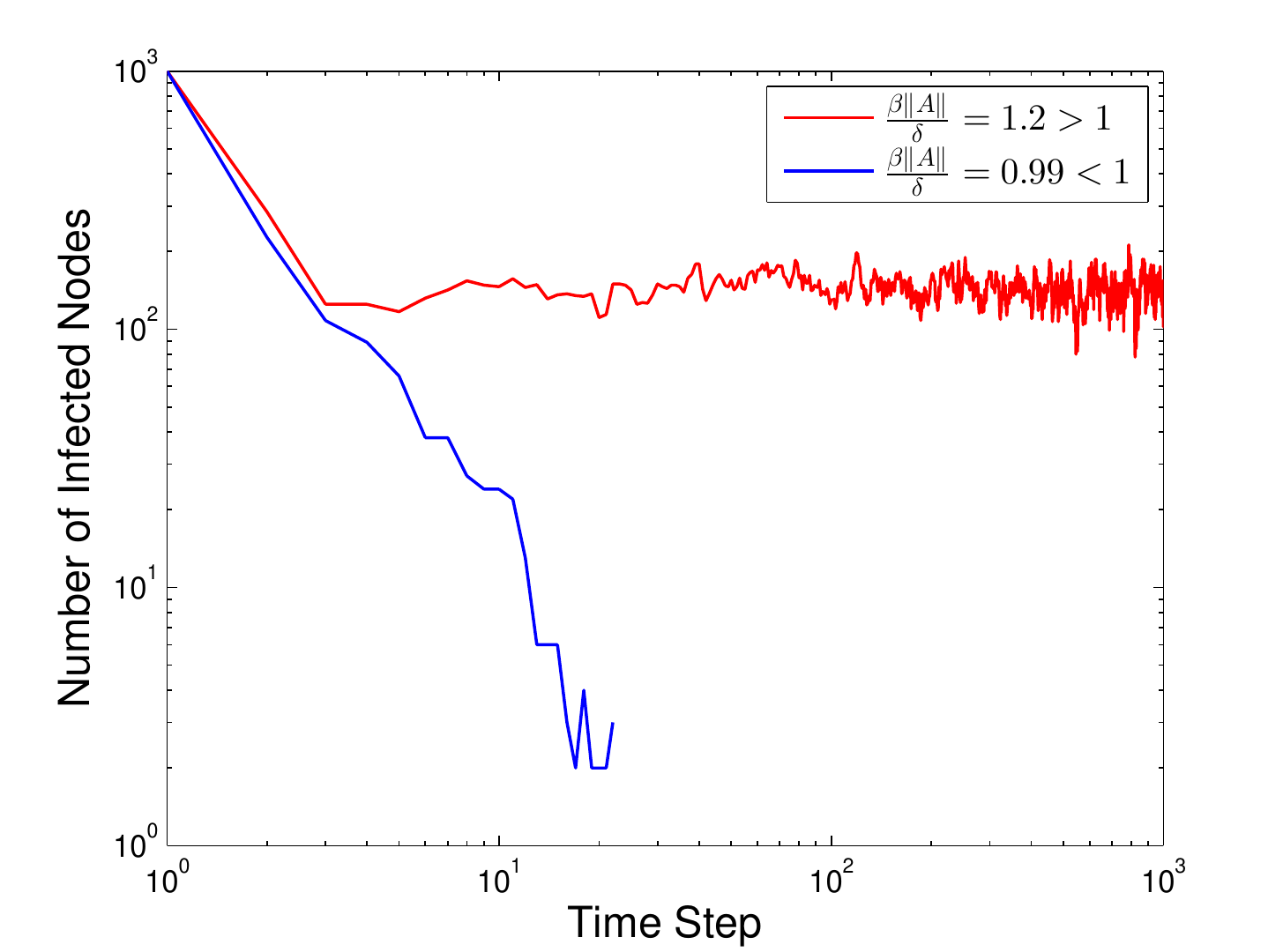}%
\label{plot1}}
\subfloat{\includegraphics[width=2.3in]{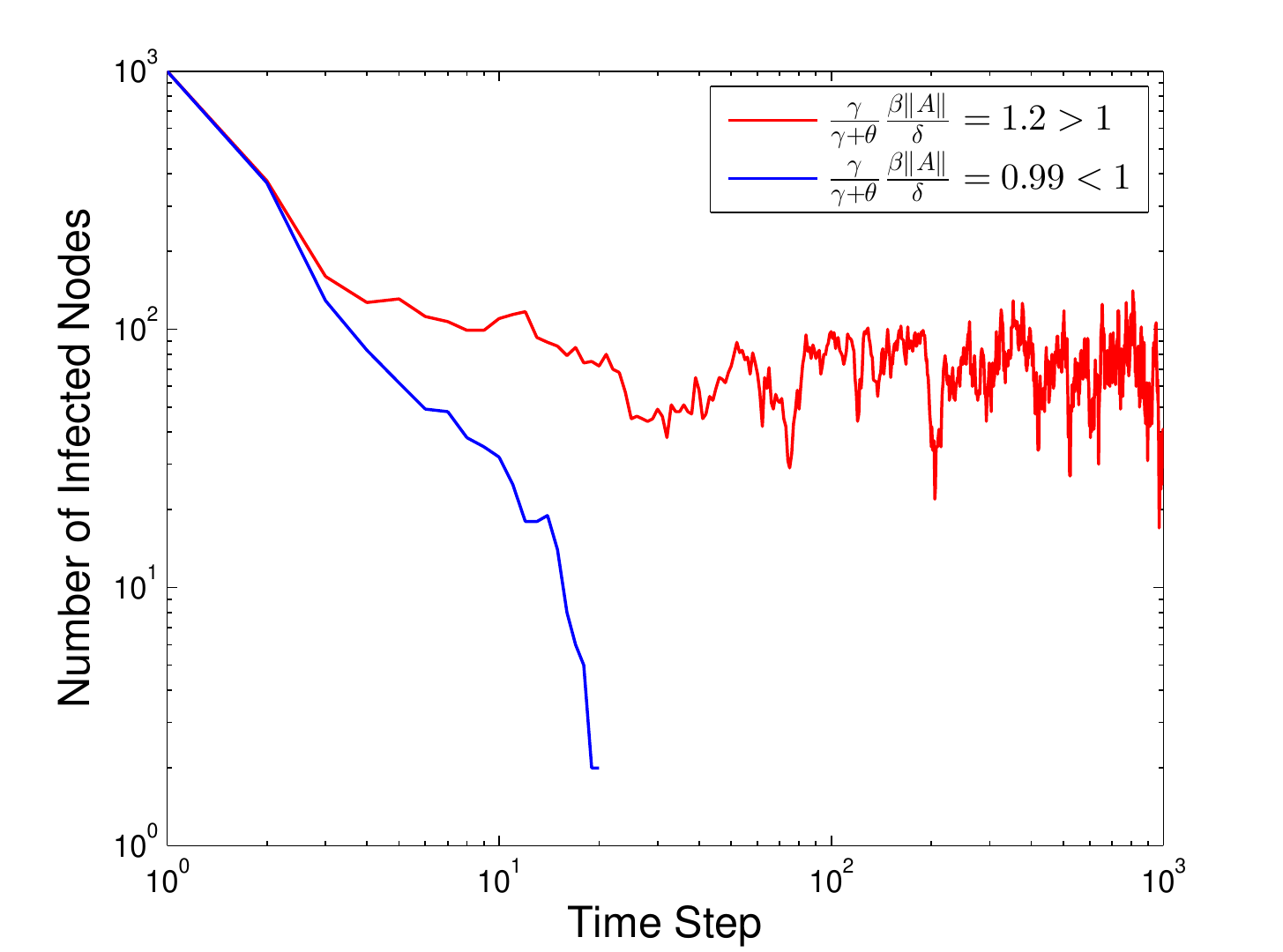}%
\label{plot2}}
\subfloat{\includegraphics[width=2.3in]{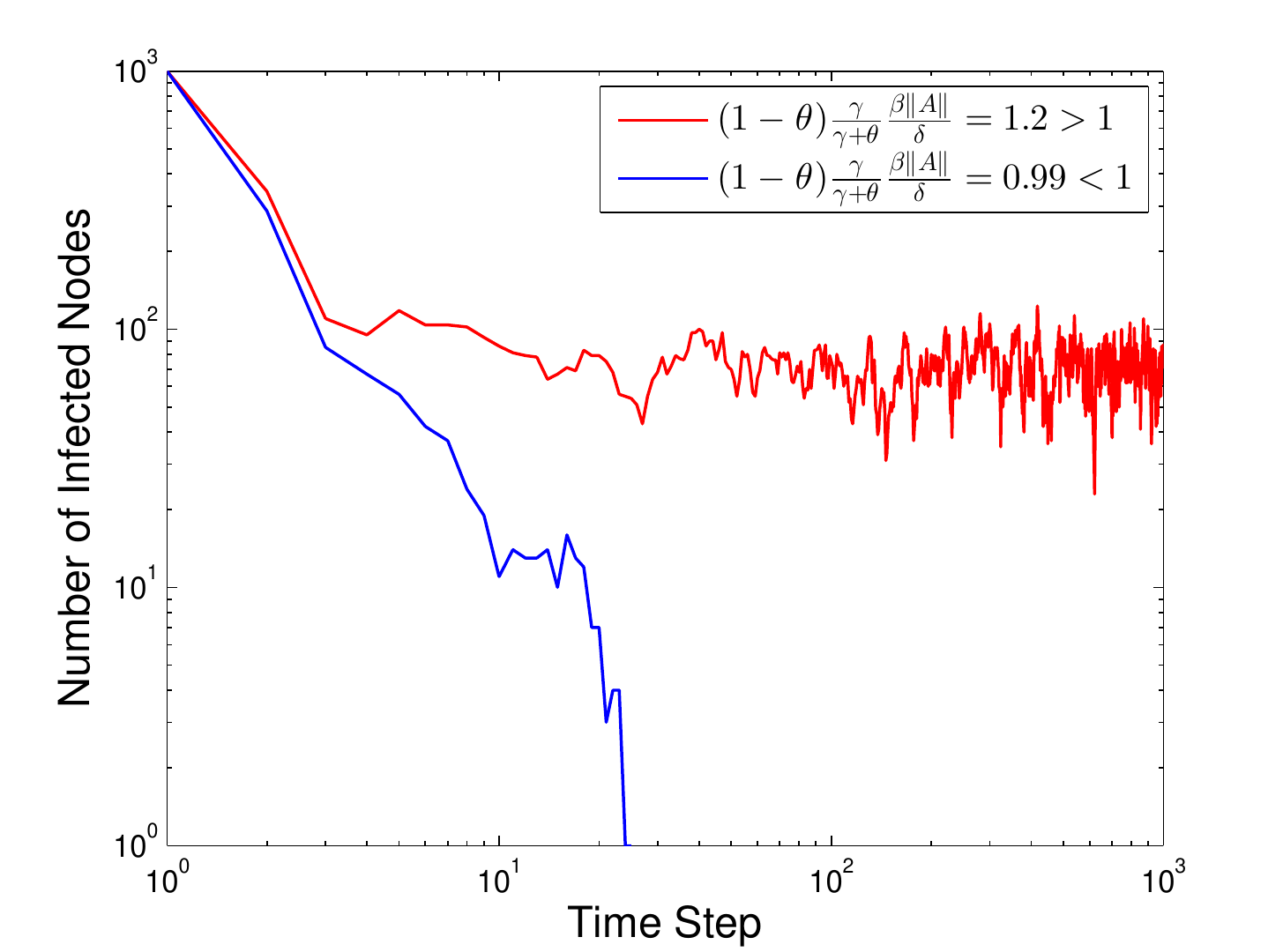}%
\label{plot3}}
\caption{The evolution of a) SIRS, b) SIV-Vaccination-Dominant, c) SIV-Infection-Dominant epidemics over an Erd\H{o}s-R\'enyi graph with $n=2000$ nodes. The blue curves show fast extinction of the epidemic. The red curves show epidemic spread around the nontrivial fixed point (convergence is not observed.)}
\label{joint}
\end{figure*}

\paragraph{Epidemic Spread: $(1-\theta)\frac{\gamma}{\gamma+\theta}\frac{\beta}{\delta}\lambda_{max}(A)>1$}
As before, the disease-free fixed point of the mapping is not stable when $(1-\theta)\frac{\gamma}{\gamma+\theta}\frac{\beta}{\delta}\lambda_{max}(A)>1$, and there exists a unique second fixed point.
\begin{theorem}
If $(1-\theta)\frac{\gamma}{\gamma+\theta}\frac{\beta}{\delta}\lambda_{max}(A)>1$, the nonlinear map (\ref{nonlinear_R_vd}, \ref{nonlinear_I_vd}), has a unique nontrivial fixed point.
\end{theorem}
The proof is similar to that of Theorem \ref{thm:SIV_id_nontrivial}, and is omitted for brevity.

\subsubsection{Analysis of the Exact Markov Chain}
As shown above, the stability condition of the main fixed point (epidemic extinction) is relaxed by a factor of $(1-\theta)$ in the vaccination-dominant model. In this part, we show that the condition for the fast mixing time of the Markov chain is also relieved by the same factor.
\begin{theorem}\label{thm_mixing_vd}
If $(1-\theta)\frac{\beta\lambda_{\max}(A)}{\delta}<1$, the mixing time of the Markov chain whose transition matrix $S$ is described by Eqs. \eqref{MC1_vd} and \eqref{MC2_vd} is $O(\log n)$.
\end{theorem}

\section{Experimental Results}

We show the simulation results on Erd\H{o}s-R\'enyi graphs, for below and above the epidemic thresholds, and they confirm the theorems proved in the paper. 

In a graph with $n=2000$ and $\lambda_{\max}(A)=16.159$, for SIS epidemics, we fix $\delta=0.9$ and try different values of $\beta$. As it can be seen in Fig. \ref{plot0}, when the condition $\frac{\beta\|A\|}{\delta}<1$ is satisfied (e.g. $\beta=0.055$) the epidemic decays exponentially, and dies out quickly. In contrast when $\frac{\beta\|A\|}{\delta}>1$ (e.g. $\beta=0.056$), the epidemic does not exhibit convergence to the disease-free state in any observable time. In fact the epidemic keeps spreading, around the nontrivial fixed point.

The same behavior is observed for the other models as well. The results are plotted in Fig. \ref{joint}, in log-log scale. $\gamma$ and $\theta$ are set to $0.5$, and we change $\beta$. For SIRS model, the threshold condition is $\frac{\beta\|A\|}{\delta}< 1$, which is the same as that of SIS, and it means having an additional recovered state does not necessarily make the system more stable.
For the first SIV model (infection-dominant), we observe the same exponential decay when $\frac{\gamma}{\gamma+\theta}\frac{\beta\|A\|}{\delta}<1$ (e.g. when $\|A\|=16.232$ and $\beta=0.11$), which means the vaccination indeed makes the system more stable. Furthermore, for the vaccination-dominant model, under $(1-\theta)\frac{\gamma}{\gamma+\theta}\frac{\beta\|A\|}{\delta}<1$ (e.g. $\beta=0.22$), we observe the fast convergence again, which confirms that the system is even more stable when vaccination is dominant. As plots show, for above the threshold cases (e.g. $\beta=0.07$ for SIRS, $0.13$ for SIV-infection-dominant, and $0.29$ for SIV-vaccination-dominant) we do not observe epidemic extinction in any reasonable time.

\section{Summary and Conclusions}

We studied the exact network-based Markov chain Model for the SIS, SIR, SIRS and SIV epidemics, and their celebrated mean-field approximations, as well as their linear approximations. 
Below a certain threshold, the disease-free fixed point is globally stable for the nonlinear model, and also the mixing time of the exact Markov chain is $O(\log n)$, which means the epidemic dies out fast. Furthermore, above a threshold, the disease-free fixed point is not stable for the linear and nonlinear models, and there exists a second unique fixed point, which corresponds to the endemic state. This nontrivial fixed point is also stable in most cases. Fig. \ref{fig:comp} compares and summarizes all the results. As one can see, for SIS and SIRS cases there is no gap between the two thresholds, but there is a gap in SIV cases, over which only the local stability of the mean-field approximation is known.
Finally we should remark that the exact epidemic threshold of the Markov chain, and whether such threshold exists, is still an open question. Extensive numerical simulations suggest the existence of such threshold and a phase transition behavior. However, the observed threshold, for certain networks, is different from the threshold for stability of the nonlinear model.

\appendix  

%

\setcounter{equation}{0}
\makeatletter
\renewcommand{\theequation}{A.\arabic{equation}}
\renewcommand{\thefigure}{A.\arabic{figure}}

{\noindent \bfseries Proof of Lemma \ref{Lem:ccv}.}
$h_{i,u,v}(s)$ is concave by property (c). 
\begin{align}
&\frac{d}{ds}\left( \frac{h_{i,u,v}(s) - h_{i,u,v}(0)}{s} \right) \nonumber \\
&= \frac{1}{s}\left(h_{i,u,v}'(s) - \frac{h_{i,u,v}(s) - h_{i,u,v}(0)}{s} \right) \\
&= \frac{1}{s}\left(h_{i,u,v}'(s) - h_{i,u,v}'(s^*) \right) \leq 0
\end{align}
$\displaystyle \frac{h_{i,u,v}(s) - h_{i,u,v}(0)}{s} = h_{i,u,v}'(s^*)$ for some $s^* \in (0,s)$ by the mean value theorem.
\hfill$\blacksquare$

{\noindent \bfseries Proof of Lemma \ref{lm:v}.}
Suppose that $\lambda_{max}((1-\delta)I_n + \beta A ) > 1 $, and $w$ is an eigenvector corresponding to the maximum eigenvalue.
$(1-\delta)I_n + \beta A$ is non-negative and irreducible because $A$ is the adjacency matrix of a connected graph $G$ 
(a non-negative matrix $X$ is irreducible if there exists $m(i,j) \in \mathbb{N}$ for each pair of indices $i,j$ such that $(X^{m(i,j)})_{i,j}$ is nonzero).
Every entry of $w$ is positive by Perron-Frobenius theorem for irreducible matrices, 
and $(\beta A - \delta I_n )w \succ 0_n$ because the eigenvalue corresponding to $w$ is greater than unity.

Suppose that there exits $v \succ 0_n$ such that $(\beta A - \delta I_n )v \succ 0_n$. Then, $((1-\delta)I_n + \beta A )v \succ v $
\begin{align}
\lambda_{max}((1-\delta)I_n + \beta A ) &= \sup_{u \in \mathbb{R}^n}  \frac{ \| ((1-\delta)I_n + \beta A )u \|_2 }{\| u \|_2} \\
&\geq \frac{ \| ((1-\delta)I_n + \beta A )v \|_2 }{\| v \|_2} > 1
\end{align}
\hfill$\blacksquare$

{\noindent \bfseries Proof of Theorem \ref{Thm:existence}.}
$U_i$ and $U$ are defined by $\Psi$ as below.
\begin{equation}
U_i = \{ x \in [0,1]^n : \Psi_i(x) \geq 0 \} \qquad U= \bigcap_{i=1}^n U_i
\end{equation}
By the lemma above, there exists $v \succ 0_n$ such that $(\beta A - \delta I_n )v \succ 0_n$. 
There is a small $\epsilon >0$ such that $\epsilon v \in U$ 
because the Jacobian of $\Psi=(\Psi_1, \cdots, \Psi_n)^\mathrm{T}$ is equal to $\beta A - \delta I_n$ at the origin and $\Psi(0)=0$ by property (a) of $\Xi$ and (d) of $\omega$.

Define $\max(x,y)= (\max(x_1,y_1), \cdots, \max(x_n,y_n))$. We claim that $\max(x,y) \in U$ if $x,y \in S$. The proof follows. 

$\max(x_i, y_i)=x_i$ without loss of generality for $x,y \in U$.
\begin{equation}
\Psi_i(\max(x,y)) = \Xi_i(\max(x,y)) - \omega(x_i) \geq \Xi_i(x) - \omega(x_i) \geq 0
\end{equation}
The first inequality holds by property (b), and the second inequality holds because $x \in U$.
Therefore $\max(x,y) \in U_i$ for every $i$ and it completes the proof of the claim.

This leads to the existence of a unique maximal point $x^* \in U$ such that $x^* \succeq x$ for all $x \in U$.
$\epsilon v \in U$ and the maximality of $x^*$ guarantees that $x^*$ has positive entries.

We claim that $\Psi_i(x^*)=0$ for all $i \in \{1, \cdots, n \}$. Assume that $\Psi_i(x^*) \neq 0$ for some $i$. Then, $\Psi_i(x^*)>0$ since $x^* \in U$. 
There exists $z_i > x_i^*$ such that 
\begin{equation}
\Psi_i(x^*) = \Xi_i(x^*) - \omega(x^*_i) > \Xi_i(x^*) - \omega(z_i) \geq 0 \label{Eq:mpz1}
\end{equation}
Define $z=(z_1, \cdots, z_n)^\mathrm{T}$ with $z_j =x^*_j$ for $j \neq i$. For every $k \in \{1, \cdots, n\}$,
\begin{equation}
\Psi_k(z) = \Xi_k(z) -\omega(z_k) \geq \Xi_k(x^*) - \omega(z_k) \geq 0 \label{Eq:mpz2}
\end{equation}
The first inequality of \eqref{Eq:mpz2} holds by property (b). The second inequality of \eqref{Eq:mpz2} holds by \eqref{Eq:mpz1} if $k=i$ and the inequality holds by definition of $z$ if $k \neq i$. \eqref{Eq:mpz2} guarantees that $z \in U$. 
$z_i > x_i^*$ and $z_j =x^*_j$ for $j \neq i$ contradict that $x^*$ is the maximal point of $U$. 
The assumption was therefore wrong, $\Psi_i(x^*)=0$ for all $i \in \{1, \cdots, n \}$, and there exists a nontrivial zero of $\Psi$.

The next step is showing that $x^*$ is the unique nontrivial zero of $\Psi$. Assume that $y^*$ is another nontrivial zero. Then $y^* \in U$ and $\Psi(y^*)=0_n$. 

We claim that every entry of $y^*$ is positive. Define $K_0$ and $K_+$ where $y^*_i=0$ if $i \in K_0$ and $y^*_i >0 $ if $i \in K_+$. Then, $K_0 \cup K_+ = \{1, \cdots, n\}$. $K_0$ and $K_+$ are separation of vertex set of the system.
Assume that $K_0$ is a non-empty set. 
There exists $j \in K_+$ such that $j$ is connected to a node in $K_0$ because $G$ is connected.
Denote $k \in K_0$ as a node which is connected to $j$.
\begin{equation}
\Psi_k(y^*) = \Xi_k(y^*) - \omega(y^*_k) = \Xi_k(y^*) > 0 
\end{equation}
The inequality above is strict by property (b) since $k \in N_j$ and $y^*_j >0$. 
It contradicts that $\Psi(y^*)=0$. $K_0$ is the empty set.

We get the following inequality by Lemma~\ref{Lem:ccv} for $u=0_n, v=x^*$ and $s \leq 1$.
\begin{equation}
\frac{\Xi_i(sx^*)}{s} = \frac{h_{i,u,v}(s) - h_{i,u,v}(0)}{s}
\geq h_{i,u,v}(1) - h_{i,u,v}(0)= \Xi_i(x^*) \label{Eq:concieq}
\end{equation}

There exists $\alpha \in (0,1)$ such that $y^* \succeq \alpha x^*$ and $y^*_j = \alpha x^*_j$ for some $j \in \{1, \cdots, n \}$.
\begin{align}
\Psi_j(y^*) &= \Xi_j(y^*) - \omega(\alpha x^*_j) \\
&\geq \Xi_j(\alpha x^*) - \omega(\alpha x^*_j) \label{Eq:ieq1} \\ 
&\geq \alpha \Xi_j(x^*) - \omega(\alpha x^*_j) \label{Eq:ieq2} \\
&> \alpha \left( \Xi_j(x^*) - \omega(x^*_j) \right) = 0 \label{Eq:ieq3} 
\end{align}
\eqref{Eq:ieq1} and \eqref{Eq:ieq2} are guaranteed by property (b) and \eqref{Eq:concieq}.
\eqref{Eq:ieq3} also holds because $\displaystyle \frac{\alpha \omega(x^*_j)}{\alpha x^*_j} > \frac{\omega(\alpha x^*_j)}{\alpha x^*_j}$
by $\alpha \in (0,1)$, $x^*_j >0$ and property (f).

This contradicts that $\Psi_i(y^*)=0$ for all $i$. Therefore $x^*$ is the unique nontrivial zero of $\Psi$.
\hfill$\blacksquare$

{\noindent \bfseries Proof of Theorem \ref{Thm:GSofNFP}.}
It is trivial to check that $\displaystyle \frac{ \partial \Phi_i }{ \partial x_j} \geq 0$ for any $i,j \in \{1,\cdots, n\}$.

Suppose that $\Phi(x) \preceq x$. 
Then, $\Phi(\Phi(x)) \preceq \Phi(x)$ since $\Phi$ is increasing.
Similarly, $\Phi(\Phi(x)) \succeq \Phi(x)$ if $\Phi(x) \succeq x$.

Define a sequence $y^{(0)}= 1_n = (1,1, \cdots, 1)^\mathrm{T}$ and $y^{(k+1)}=\Phi(y^{(k)})$. 
\begin{equation}
y^{(1)} = (1-\delta)1_n + \delta \Xi(1_n) \preceq 1_n = y^{(0)} 
\end{equation}
The equation above implies that $y^{(k+1)} \preceq y^{(k)}$ for every $k \in \mathbb{N}$.
The sequence $\{y^{(k)}\}_{k=0}^\infty \subset [0,1]^n$ has a limit point because it is decreasing, and bounded from below.
Denote $y^*$ as a limit point of the sequence, then $\Phi(y^*)=y^*$. 
There are two candidates for $y^*$ because $\Phi$ has only two fixed points.

Since $\Phi$ is an increasing map, and $y^{(0)} \succeq x$ for every $x \in [0,1]^n$, $y^{(k)} \succeq \Phi^k(x)$.
$y^{(k)} \succeq \Phi^k(x^*)=x^*$ for every $k$ implies that $y^* \succeq x^*$. It also implies that $y^*=x^*$.
For any $P(0) \in [0,1]^n$, an upper bound of $P(t)$ is $y^{(t)}$ and it goes to $x^*$ as $t$ goes to infinity.

Suppose that all the entries of $P(0)$ are positive. 
This is reasonable since there exists $m$ such that all the entries of $P(m)$ are positive if $P(0)$ is not the origin.
There exists $\alpha \in (0,1)$ such that $\alpha x^* \preceq P(0)$.
Define a sequence $z^{(0)}= \alpha x^*$ and $z^{(k+1)}=\Phi(z^{(k)})$. 
\begin{align}
z^{(1)}_i &= z^{(0)}_i + ( 1- (1-\delta)z^{(0)}_i)\left( \Xi_i(\alpha x^*) - \frac{\delta \alpha x^*_i}{1- (1-\delta)\alpha x^*_i} \right) \\
&> z^{(0)}_i + \alpha ( 1- (1-\delta)z^{(0)}_i)\left( \Xi_i(x^*) - \frac{\delta x^*_i}{1- (1-\delta)x^*_i} \right)\\
&= z^{(0)}_i
\end{align}
The inequality above holds by \eqref{Eq:ieq1}, \eqref{Eq:ieq2} and \eqref{Eq:ieq3}.
It implies that $z^{(k+1)} \succeq z^{(k)}$ for every $k \in \mathbb{N}$, and $z^{(k)}$ gives a lower bound for $P(k)$. 
Since $z^{(0)}= \alpha x^* \preceq x^*$, $z^{(k)} \preceq \Phi^k(x^*)=x^*$. 
$\{z^{(k)}\}_{k=0}^\infty \subset [0,1]^n$ has a limit point because it is increasing, and bounded from above.
$x^*$ is the only possible limit point of $\{z^{(k)}\}_{k=0}^\infty$. 
The lower bound of $P(t)$ is $z^{(t)}$ and it goes to $x^*$ as $t$ goes to infinity.

Both the upper and lower bounds of $P(t)$ go to $x^*$ which implies that $P(t)$ converges to $x^*$.
\hfill$\blacksquare$

{\noindent \bfseries Proof of Proposition \ref{Lem:LP}.}
We maximize the marginal probability of infection at time $t+1$ ($p_i(t+1)$) given marginals at time $t$ ($p_j(t)$'s). For the sake of simplicity, let us drop the time index $t$ and mark time index only for $t+1$.
\begin{align}
\max_{\mu B = p^T, \mu \succeq 0} p_i(t+1) 
&= \max_{\mu B = p^T, \mu \succeq 0} \mu S B f_i \\ 
&= \max_{\mu \succeq 0} \min_\lambda \mu S B f_i - (\mu B - p^T)\lambda \\
&= \min_\lambda \max_{\mu \succeq 0} \mu (S B f_i - B \lambda ) + p^T \lambda 
\end{align}
$\max_{\mu \succeq 0} \mu (S B f_i - B \lambda) = +\infty$ if any entry of $(S B f_i - B\lambda)$ is strictly positive. It follows that $ S B f_i - B \lambda \preceq 0$. Evaluation of $S B f_i$ and $ B \lambda $ is as follows.
\begin{align}
(S B f_i)_X &=(S B)_{X,i} = \sum_{Y \in \{0,1\}^n} S_{X,Y} B_{Y,i} = \sum_{Y \in \{0,1\}^n} S_{X,Y} Y_i \\
&= \mathbb{P}(Y_i =1|X)= \left\{
\begin{array}{rl}
1-(1-\beta)^m & \text{if } X_i=0,\\
1-\delta(1-\beta)^m & \text{if } X_i=1.
\end{array} \right. \label{Eq:SBf_i}
\end{align}
$\mathbb{P}(Y_i=1|X)$ follows \eqref{Eq:noimmune} and $m$ is the number of infected neighbors of $i$ as stated before.
\begin{equation}
(B \lambda)_X = \lambda_0 + \sum_{k=1}^n B_{X,k} \lambda_k = \lambda_0 + \sum_{k=1}^n \lambda_k X_k \label{Eq:Blambda}
\end{equation}

We try several $X$ for \eqref{Eq:SBf_i}, \eqref{Eq:Blambda} and $ S B f_i - B \lambda \preceq 0$ to get a feasible $\lambda$.
\begin{equation}
\left\{
\begin{array}{rl}
X=\bar{0} &,  \lambda_0 \geq 0 \\
X=\hat{i} &,  \lambda_0 + \lambda_i \geq 1-\delta \\
X=\hat{j} , j \in N_i &, \lambda_0 + \lambda_j \geq \beta \\
X=\hat{j} , j \notin N_i &, \lambda_0 + \lambda_j \geq 0 \end{array} \right. \label{Eq:feasibility}
\end{equation}

We claim that $\lambda^*=(\lambda^*_0, \lambda^*_1, \cdots, \lambda^*_n )^T$ defined by $\lambda^*_0 = 0$, $\lambda^*_i = 1-\delta$, $\lambda^*_j= \beta$ for $j \in N_i$ and $\lambda^*_j= 0$ for $j \notin N_i$ is in feasible set.

For $X_i=0$, $|N_i \cap \mathbb{S}(X)|=m$
\begin{equation}
\mathbb{P}(Y_i =1|X) = 1-(1-\beta)^m \leq m\beta = \lambda^*_0 + \sum_{k=1}^n \lambda^*_k X_k
\end{equation}

For $X_i=1$, $|N_i \cap \mathbb{S}(X)|=m$
\begin{equation}
\mathbb{P}(Y_i =1|X) = 1-\delta(1-\beta)^m \leq 1- \delta + m\beta = \lambda^*_0 + \sum_{k=1}^n \lambda^*_k X_k
\end{equation}

Therefore $\lambda^*$ is in feasible set. 
\begin{align}
\max_{\mu B = p^T, \mu \succeq 0} p_i(t+1) 
&= \min_\lambda \max_{\mu \succeq 0} \mu (S B f_i - B \lambda ) + p^T \lambda \\
&\leq p^T \lambda^* = (1-\delta)p_i + \beta \sum_{j \in N_i} p_j
\end{align}
\hfill$\blacksquare$

{\noindent \bfseries Proof of Theorem \ref{Thm:upperboundmt}.}
In order to compute the mixing time \eqref{Eq:mixingdefn}, we should find the supremum of $\|\mu S^t-\pi\|_{TV}$. We have
\begin{align}
\|\mu S^t-\pi\|_{TV} &= \frac{1}{2}\sum\limits_X \lvert(\mu S^t)_X - \pi_X\rvert\\
&= \frac{1}{2}\sum\limits_X \lvert(\mu S^t)_X - (e_{\bar{0}})_X\rvert\\
&= \frac{1}{2}\big( 1-(\mu S^t)_{\bar{0}}\big) + \frac{1}{2}\sum\limits_{X\neq \bar{0}} (\mu S^t)_X\\
&= \frac{1}{2}\big( 1-(\mu S^t)_{\bar{0}}\big) + \frac{1}{2}\big( 1-(\mu S^t)_{\bar{0}}\big)\\
&= 1-(\mu S^t)_{\bar{0}}\\
&\leq 1-(e_{\bar{1}} S^t)_{\bar{0}} .
\end{align}
In the last inequality we have used the fact that the worst-case initial $\mu$ is the all-infected state, i.e. $e_{\bar{1}}$. This is rigorously established in Section \ref{sec:partialordering} through the partial ordering. 

Now for any $t< t_{mix}(\epsilon)$ we have
\begin{align}
\epsilon &< 1-(e_{\bar{1}} S^t)_{\bar{0}}\\
 &=1-\mathbb{P}\left( \substack{\text{all nodes are healthy at time $t$} \mid\quad \\\text{all nodes were infected at time $0$}}\right)\\
 &=\mathbb{P}\left( \substack{\text{some nodes are infected at time $t$} \mid \\\text{all nodes were infected at time $0$}} \right)\\
 &\leq \sum\limits_{i=1}^n \mathbb{P}\left( \substack{\text{node $i$ is infected at time $t$} \mid\qquad \\\text{all nodes were infected at time $0$}} \right)\\
 &\leq 1_{n}^T ((1-\delta)I_n + \beta A)^t 1_{n}\label{epsilon}\\
 &\leq \|1_{n}\|^2 \|(1-\delta)I_n + \beta A\|^t\\
 &=n\|(1-\delta)I_n + \beta A\|^t .\label{epsilon_last}
\end{align}
where \eqref{epsilon} comes from the fact that $\sum\limits_{i=1}^n p_i(t) = 1_{n}^T p(t) \leq 1_{n}^T ((1-\delta)I_n + \beta A)^t p(0)$.

Now since $\|(1-\delta)I_n + \beta A\|<1$, we get $t< \frac{\log \frac{n}{\epsilon}}{-\log \|(1-\delta)I_n + \beta A\|}$ for all $t< t_{mix}(\epsilon)$. Therefore $t_{mix}(\epsilon)\leq \frac{\log \frac{n}{\epsilon}}{-\log \|(1-\delta)I_n + \beta A\|}$, which concludes that the mixing time is $O(\log n)$.
\hfill$\blacksquare$

{\noindent \bfseries Proof of Lemma \ref{Lem:RinverseSR}.}
We want to compute the inverse matrix of $R$ first. Define a matrix $R'$.
\begin{equation}
R'_{X,Y}= \left\{
\begin{array}{rl}
(-1)^{|\mathbb{S}(Y-X)|} & \text{if } X \preceq Y,\\
0 & \text{otherwise } 
\end{array} \right.
\end{equation}
$|\mathbb{S}(Y-X)|$ represents the number of nodes which are infected in $Y$, but not in $X$. We claim that $R'=R^{-1}$. If $X \npreceq Y$, then $X \npreceq Z$ or $Z \npreceq Y$ holds for every $Z \in \{0,1\}^n$. By the definition of $R$ and $R'$, $R_{X,Z}=0$ or $R'_{Z,Y}=0$ if $X \npreceq Y$. It is straightforward that $(RR')_{X,Y}=0$ if $X \npreceq Y$. It's enough to consider the case $X \preceq Y$.
\begin{align}
(RR')_{X,Y} &= \sum_Z R_{X,Z} R'_{Z,Y} = \sum_{X \preceq Z \preceq Y} 1^{|\mathbb{S}(Z-X)|} (-1)^{|\mathbb{S}(Y-Z)|} \nonumber \\
&= (1-1)^{|\mathbb{S}(Y-X)|}
\end{align}
$(RR')_{X,Y} = 1$ if $|\mathbb{S}(Y-X)|=0$ and $(RR')_{X,Y} = 0$ otherwise. It leads that $RR'$ is an identity matrix of size $2^n$ and $R'=R^{-1}$.
\begin{align}
&(R^{-1}SR)_{X,Z} \\
&=\sum_{Y \preceq Z} (R^{-1} S)_{X,Y} = \sum_{Y \preceq Z} \sum_W R^{-1}_{X,W} S_{W,Y} \\
&= \sum_{Y \preceq Z} \sum_{W \succeq X} (-1)^{|\mathbb{S}(W-X)|} S_{W,Y} \\
&=\sum_{W \succeq X} (-1)^{|\mathbb{S}(W-X)|} \prod_{i \in \mathbb{S}(Z)^c} \mathbb{P}(\xi_i(t+1)=0|\xi(t)=W) \\
&=\sum_{W \succeq X} (-1)^{|\mathbb{S}(W-X)|} \delta^{|\mathbb{S}(W) \cap \mathbb{S}(Z)^c|}(1-\beta)^{\sum_{i\in \mathbb{S}(Z)^c } | N_i \cap \mathbb{S}(W) |}\\
&=\sum_{W \succeq X} (-1)^{|\mathbb{S}(W-X)|} \delta^{|\mathbb{S}(W) \cap \mathbb{S}(Z)^c|}(1-\beta)^{\sum_{i\in \mathbb{S}(W) } | N_i \cap \mathbb{S}(Z)^c |} \label{Eq:RinverseSR1}
\end{align}

By some algebra,
\begin{align}
&\delta^{-|\mathbb{S}(X) \cap \mathbb{S}(Z)^c|}(1-\beta)^{-\sum_{i \in \mathbb{S}(X) } | N_i \cap \mathbb{S}(Z)^c |}(R^{-1}SR)_{X,Z} \\
&= \sum_{W \succeq X} (-1)^{|\mathbb{S}(W-X)|} \delta^{|\mathbb{S}(W-X) \cap \mathbb{S}(Z)^c|}(1-\beta)^{\sum_{i\in \mathbb{S}(W-X) } | N_i \cap \mathbb{S}(Z)^c |} \\
&= \prod_{i \in \mathbb{S}(X)^c} \left( 1 - (1-\beta)^{| N_i \cap \mathbb{S}(Z)^c |} \delta^{1_{\{i \in \mathbb{S}(Z)^c\}}} \right) \label{Eq:RinverseSR2}
\end{align}

Define $\neg X=\bar{1} - X$. $\neg X$ is an opposite state of $X$ where each node is healthy in $\neg X$ if it is infected in $X$ and vice versa. From \eqref{Eq:RinverseSR1} and \eqref{Eq:RinverseSR2}, We simplify $(R^{-1}SR)_{X,Z}$ using $\neg X$ and $\neg Z$.
\begin{equation}
(R^{-1}SR)_{X,Z}=\mathbb{P}(\xi(t+1)=\neg X|\xi(t)=\neg Z) \geq 0
\end{equation}
\hfill$\blacksquare$

{\noindent \bfseries Proof of Lemma \ref{Lem:order}.}
We defined $2^n$-dimensional square matrix $R$ from Lemma~\ref{Lem:RinverseSR} because we can represent $\mu \leq_{st} \mu'$ using $R$. By definition of $\mu \leq_{st} \mu'$,
\begin{equation}
((\mu-\mu')R)_Y = \sum_X (\mu-\mu')_X R_{X,Y} = \sum_{X \preceq Y} (\mu - \mu')_X \geq 0 \label{Eq:RYpositive}
\end{equation}
$\mu \leq_{st} \mu'$ if and only if all of $(\mu-\mu')R$'s entries are non-negative. $((\mu-\mu')R)_Y=0$ if $Y=\bar{1}=(1, 1,, \cdots, 1)$ because both of $\mu$ and $\mu'$ are probability vectors whose $1$-norm is $1$. 

Define a row vector $\nu \in \mathbb{R}^{\{0,1\}^n}$ whose $Y$-th element is defined by $\nu_Y=((\mu-\mu')R)_Y$. $\nu_Y \geq 0$ for all $Y \in \{0,1\}^n$ by \eqref{Eq:RYpositive}. $\nu$ is a non-negative row vector, and $\nu_{\bar{1}}=0$. $\mu-\mu'=\nu R^{-1}$. We can understand $\mu-\mu'$ as a conical combination of all row vectors of $R^{-1}$ but the $\bar{1}$-th row vector. 

$\mu S \leq_{st} \mu' S$ if and only if $(\mu-\mu')SR$ is a non-negative vector. $\mu-\mu'=\nu R^{-1}$ for non-negative $\nu$ since $\mu \leq_{st} \mu'$. $(\mu-\mu')SR=\nu R^{-1}SR$ is non-negative since $\nu$ is non-negative and $R^{-1}SR$ is a matrix all of whose entries are non-negative by Lemma~\ref{Lem:RinverseSR}.
\hfill$\blacksquare$

{\noindent \bfseries Proof of Lemma \ref{Lem:uofr}.}
We begin by evaluating each entry of $Su(r)$.
\begin{align}
&(S u(r))_X \\
&= \sum_{Y \in \{0,1\}^n} S_{X,Y} u(r)_Y  \\
&= \sum_{Y \in \{0,1\}^n} \Big( \prod_{i \in \mathbb{S}(Y)} (1-r_i) \mathbb{P}(Y_i=1|X) \Big) \Big( \prod_{i \notin \mathbb{S}(Y)} \mathbb{P}(Y_i=0|X) \Big) \\
&=\prod_{i=1}^n (1-r_i) \mathbb{P}(Y_i=1|X) + \mathbb{P}(Y_i=0|X)
\end{align}

Assume $\mathbb{S}(X) \cap \mathbb{S}(Z) = \emptyset$ for two states $X,Z \in \{0,1\}^n$ i.e. there is no common infected node in the two states $X$ and $Z$. It is trivial to check that the following is true:
\begin{equation}
\mathbb{P}(Y_k=0|X+Z)=\mathbb{P}(Y_k=0|X)\mathbb{P}(Y_k=0|Z)
\end{equation}

For simplicity, we call $q_{k,X}=\mathbb{P}(Y_k=0|X)$.
\begin{align}
&(S u(r))_{X+Z} \\
&=\prod_{i=1}^n (1-r_i) \mathbb{P}(Y_i=1|X+Z) + \mathbb{P}(Y_i=0|X+Z) \\
&=\prod_{i=1}^n (1-r_i)(1-q_{i,X+Z}) + q_{i,X+Z} \\
&=\prod_{i=1}^n (1-r_i)(1-q_{i,X}q_{i,Z}) + q_{i,X}q_{i,Z} \\
&\geq \prod_{i=1}^n \left( (1-r_i)(1-q_{i,X}) + q_{i,X} \right) \left( (1-r_i)(1-q_{i,Z}) + q_{i,Z} \right) \label{Eq:dividingSur1} \\
&= (S u(r))_X (S u(r))_Z \label{Eq:dividingSur2}
\end{align}

\eqref{Eq:dividingSur1} holds by the following one for $a,b,c \in [0,1]$:
\begin{align}
&(c(1-ab)+ab) - (c(1-a) + a)(c(1-b) + b) \nonumber \\
&= c(1-c)(1-a)(1-b) \geq 0
\end{align}

Define $\hat{i} \in \{0,1\}^n$ as the state where everyone is healthy but $i$. The following inequality holds by \eqref{Eq:dividingSur2}.
\begin{align}
(S u(r))_X & \geq \prod_{i \in \mathbb{S}(X)}(S u(r))_{\hat{i}} \\
&=\prod_{i \in \mathbb{S}(X)} \prod_{j=1}^n (1-r_j) \mathbb{P}(Y_j=1|\hat{i}) + \mathbb{P}(Y_j=0|\hat{i}) \\ \label{Eq:r'eq}
&=\prod_{i \in \mathbb{S}(X)} \left( (1-r_i)(1-\delta)+\delta \right) \prod_{j \sim i} \left( (1-r_j)\beta + 1-\beta \right) \\
&=\prod_{i \in \mathbb{S}(X)}(1-(1-\delta)r_i)\prod_{j \sim i} (1- \beta r_j) \\
&= \prod_{i \in \mathbb{S}(X)} 1- \Phi_i(r) \\
&= u(\Phi(r))_X
\end{align}
\hfill$\blacksquare$

{\noindent \bfseries Proof of Theorem \ref{thm:SIRS_nontrivial}.}
Let's define the map $\Psi \colon [0,1]^{2n} \to \mathbb{R}^n$ as $\Psi=[\Psi_1,\dots,\Psi_n]^\top$ with
\begin{equation}
\Psi_i(P(t))=\Xi_i(P_I(t))-\omega(P_{R,i}(t),P_{I,i}(t)) .
\end{equation}
Zeros of $\Psi$ correspond to fixed points of the nonlinear map (Eq. \ref{nonlinear_I_new}).

Now we define sets $U_i$ and $U$ as follows:
\begin{equation}
U_i=\{ x_I\in[0,1]^n: \Psi_i(\begin{bmatrix}x_R\\x_I\end{bmatrix})\geq 0, 0_n\preceq x_R\preceq 1_n-x_I\} ,
\end{equation}
\begin{equation}
U=\bigcap\limits_{i=1}^n U_i .
\end{equation}
In plain words, $U$ is the set of ``infection situations'' from which the system becomes ``more infected'' or remains there.

From Lemma \ref{lm:v}, $\lambda_{\max}((1-\delta)I_n+\beta A)>1$ implies that there exists $v\succ0_n$ such that $(\beta A-\delta I_n)v\succ0_n$. On the other hand $\Psi(0_{2n})=0_n$ and the Jacobian of $\Psi$ at the origin is equal to $\begin{bmatrix}0_{n\times n} & \beta A-\delta I_n\end{bmatrix}_{n\times 2n}$. As a result, there exists a small $\epsilon>0$ such that $\Psi(\begin{bmatrix}\epsilon u\\ \epsilon v\end{bmatrix})=(\beta A-\delta I_n)v\epsilon$, which is $\succ 0_n$, and indicates that $\epsilon v \in U$.

We claim that if $x,y\in U$, then $\max(x,y)\triangleq(\max(x_1,y_1),\dots,\max(x_n,y_n))\in U$.
For all $i\in\{1,\dots,n\}$, $\exists\, a_i\in[0,1-x_i]\mbox{ s.t. } \Xi_i(x)-\omega(a_i,x_i)\geq0$, and $\exists\, b_i\in[0,1-y_i]\mbox{ s.t. } \Xi_i(y)-\omega(b_i,y_i)\geq0$.
\begin{align}
\Psi_i(\begin{bmatrix}c\\ \max(x,y)\end{bmatrix}) &=\Xi_i(\max(x,y))-\omega(c_i,\max(x_i,y_i)).
\end{align}
Without loss of generality assume $\max(x_i,y_i)=x_i$, then if we pick $c_i=a_i$, it follows that:
\begin{align}
\Psi_i(\begin{bmatrix}c\\ \max(x,y)\end{bmatrix}) &=\Xi_i(\max(x,y))-\omega(a_i,x_i)\\
&\geq \Xi_i(x)-\omega(a_i,x_i)\geq 0\label{random_ineq_1} .
\end{align}
Inequality \eqref{random_ineq_1} comes from Property (b). Now $\max(x,y)\in U_i$, and we can use the same argument for all $i$. Hence $\max(x,y)\in U$, and the claim is true.

It follows that there exists a unique maximal point $x^*\in U$ such that $x^*\succeq x$ for all $x\in U$. Moreover, since $\epsilon v\in U$, we can conclude that $x^*\succ 0_n$ (all elements of $x^*$ are positive).

Now we further claim that $\Psi_i(\begin{bmatrix}a\\x^*\end{bmatrix})=0$ for some $0_n\preceq a\preceq 1_n-x^*$ and $\forall i\in \{1,\dots,n\}$. Assume, by the way of contradiction, that $\Psi_i(\begin{bmatrix}a\\x^*\end{bmatrix})\neq0$ for all $0_n\preceq a\preceq 1_n-x^*$, which means $\Psi_i(\begin{bmatrix}a\\x^*\end{bmatrix})>0$. Since $\Psi_i(\begin{bmatrix}a\\x^*\end{bmatrix})=\Xi_i(x^*)-\omega(a_i,x_i^*)>0$ and $\omega(a_i,x_i^*)<\omega(a_i,z_i)$ for any $z_i>x_i^*$ (Property (e)), there exists $z_i>x_i^*$ such that
\begin{equation}\label{eq13}
\Xi_i(x^*)-\omega(a_i,z_i)\geq 0 .
\end{equation}
Now define $z=[z_1,\dots,z_n]^\top$ with $z_j=x_j \, \forall j\neq i$. For every $k\in\{1,\dots,n\}$ we have
$$\Psi_k(\begin{bmatrix}a\\z\end{bmatrix})=\Xi_k(z)-\omega(a_k,z_k)\geq \Xi_k(x^*)-\omega(a_k,z_k)\geq 0 ,$$
for some $0_n\preceq a\preceq 1_n-z$. The first inequality holds by Property (b). The second inequality holds by \eqref{eq13} for $k=i$, and by definition for $k\neq i$. It implies that $z\in U$. Since $z_i>x_i^*$, this contradicts the fact that $x^*$ is the maximal point of $U$. Hence $\Psi_i(\begin{bmatrix}a\\x^*\end{bmatrix})=0$ for some $0_n\preceq a\preceq 1_n-x^*$, and this is true for all $i\in\{1,\dots,n\}$. Thus far we have proved that there exists a nontrivial zero for $\Psi$.

We note that in order for a point $\begin{bmatrix}p_R^*\\p_I^*\end{bmatrix}$ to be a fixed point of the nonlinear map, it should satisfy Eq. \eqref{nonlinear_R_new}, i.e.
\begin{equation}\label{relation}
p_{R,i}^*=(1-\gamma)p_{R,i}^*+\delta p_{I,i}^* \implies p_{R,i}^*=\frac{\delta}{\gamma}p_{I,i}^* .
\end{equation}

For proving the uniqueness of nontrivial zero of $\Psi$, assume by contradiction that in addition to $x^*$, $y^*$ is another nontrivial zero. Therefore $y^*\in U$, and $\Psi(\begin{bmatrix}b\\y^*\end{bmatrix})=0_n$ for some $0_n\preceq b\preceq 1_n-y^*$.

We claim that $y^*$ is all-positive. Let us define $K_0=\{1\leq i\leq n : y_i^*=0\}$ and $K_+=\{1\leq i\leq n : y_i^*>0\}$. $K_0\cup K_+=\{1,\dots,n\}$. Assume that $K_0$ is not empty. Since $G$ is connected, there exist $k\in K_0$ and $j\in K_+$ such that they are neighbors.
\begin{equation}
\Psi_k(\begin{bmatrix}b\\y^*\end{bmatrix})=\Xi_k(y^*)-\omega(b_k,y_k^*)=\Xi_k(y^*)>0 .
\end{equation}
The second equality holds by Property (d) and due to $b_k=y_k^*=0$ (from Eq. \ref{relation}). The inequality comes from Property (b) ($k\in N_j$) and $y_j^*>0$. This contradicts $\Psi(\begin{bmatrix}b\\y^*\end{bmatrix})=0_n$, and implies that $K_0=\emptyset$, and therefore every element of $y^*$ is positive.

By Property (c) and from Lemma \ref{Lem:ccv}, we know for $s\leq 1$
$$\frac{\Xi_i(u+sv)-\Xi_i(u)}{s}\geq \frac{\Xi_i(u+v)-\Xi_i(u)}{1} .$$
By setting $u=0_n$ and $v=x^*$, and using Property (a), it follows that
\begin{equation}\label{nice}
\frac{\Xi_i(sx^*)}{s}\geq \Xi_i(x^*) .
\end{equation}

For $x^*$ and $y^*$ there exists $\alpha\in(0,1)$ such that $y^*\succeq\alpha x^*$ and $y_j^*=\alpha x_j^*$ for some $j\in\{1,\dots,n\}$.
\begin{align}
\Psi_j(\begin{bmatrix}b\\y^*\end{bmatrix}) &=\Xi_j(y^*)-\omega(b_j,\alpha x_j^*)\\
&\geq \Xi_j(\alpha x^*)-\omega(b_j,\alpha x_j^*)\label{rnd_eq_1}\\
&\geq \alpha\Xi_j(x^*)-\omega(b_j,\alpha x_j^*)\label{rnd_eq_2}\\
&> \alpha\Xi_j(x^*)-\alpha \omega(\frac{b_j}{\alpha},x_j^*)\label{rnd_eq_3}\\
&= \alpha\big(\Xi_j(x^*)- \omega(a_j,x_j^*)\big)=0\label{rnd_eq_4} .
\end{align}
Inequality \eqref{rnd_eq_1} holds by Property (b), \eqref{rnd_eq_2} follows from \eqref{nice}, \eqref{rnd_eq_3} holds by Property (f), and finally \eqref{rnd_eq_4} comes from \eqref{relation}. This contradicts that $\Psi_i(\begin{bmatrix}b\\y^*\end{bmatrix})=0$ for all $i$.

It concludes that $\begin{bmatrix}a\\x^*\end{bmatrix}$ is the unique nontrivial zero of $\Psi$, and hence the unique nontrivial fixed point of the system.

\hfill$\blacksquare$

{\noindent \bfseries Proof of Theorem \ref{thm_mixing}.}
First we use the linear programming technique to show for each $i\in\left\{n+1,n+2,\dots,2n\right\}$, we have $p_i(t+1) \leq (1-\delta)p_i(t) + \beta\sum\limits_{j\in N_i}p_j(t)$.
Let $f_i\in \mathbb{R}^{2n+1}$ represent the $i^{th}$ unit column vector.
For the sake of convenience, let us drop the time index $(t)$.
\begin{align}
\max_{\mu B=p^T, \mu\succeq 0} p_i(t+1) &= \max_{\mu B=p^T, \mu\succeq 0} \mu SBf_i\\
 &= \max_{\mu\succeq 0} \min_\lambda \mu SBf_i-(\mu B-p^T)\lambda\\ \label{minmax}
 &= \min_\lambda \max_{\mu\succeq 0} \mu(SBf_i-B\lambda) +p^T\lambda ,
\end{align}
where $\lambda\in\mathbb{R}^{2n+1}$ is a column vector. If any element of $(SBf_i-B\lambda)$ is strictly positive, it leads to $\max_{\mu\succeq 0} \mu(SBf_i-B\lambda)=+\infty$. Therefore:
\begin{equation}\label{leq}
SBf_i-B\lambda\preceq 0 .
\end{equation}
Now we proceed with further calculation of $SBf_i$ and $B\lambda$.
\begin{align}
&(SBf_i)_X = (SB)_{X,i} = \sum_{Y\in\left\{0,1,2\right\}^n} S_{X,Y}B_{Y,i}\\
 &= \begin{cases} \mathbb{P}\left( Y_i=2 \mid X\right\},& i\in\left\{1,2,\dots,n\right)\\ \mathbb{P}\left( Y_{i-n}=1 \mid X\right\},& i\in\left\{n+1,n+2,\dots,2n\right)\end{cases}\\
 &= \begin{cases}
0,& \text{if } i\in\left\{1,2,\dots,n\right\} \text{ and } X_i=0\\
\delta,& \text{if } i\in\left\{1,2,\dots,n\right\} \text{ and } X_i=1\\
1-\gamma,& \text{if } i\in\left\{1,2,\dots,n\right\} \text{ and } X_i=2\\
1-(1-\beta)^{m_{i-n}},& \text{if } i\in\left\{n+1,\dots,2n\right\} \text{ and } X_{i-n}=0\\
1-\delta,& \text{if } i\in\left\{n+1,\dots,2n\right\} \text{ and } X_{i-n}=1\\ \label{SBf_i}
0,& \text{if } i\in\left\{n+1,\dots,2n\right\} \text{ and } X_{i-n}=2\\
\end{cases}\\ \label{Bl}
&(B\lambda)_X =\lambda_0 + \sum\limits_{k=1}^n B_{X,k}\lambda_k + \sum\limits_{k=n+1}^{2n} B_{X,k}\lambda_k .
\end{align}
As mentioned earlier, we want to evaluate $p_i(t+1)$ only for $i\in\left\{n+1,n+2,\dots,2n\right\}$.
Define $\hat{i} \in \left\{0,1,2\right\}^n$ as the state where only $i$ is infected, and the rest are susceptible. Trying several $X$ in \eqref{leq} using \eqref{SBf_i} and \eqref{Bl} yields:
\begin{equation}
\begin{cases}
X=\bar{0}, &\lambda_0+0+0 \geq 0\\
X=\bar{2}, &\lambda_0+\sum\limits_{k=1}^n \lambda_k+0 \geq 0\\
X=\hat{i}, &\lambda_0+0+\lambda_{n+i}\geq 1-\delta\\
X=\hat{j}, j\in N_i, &\lambda_0+0+\lambda_{n+j}\geq \beta\\
X=\hat{j}, j\not\in N_i, &\lambda_0+0+\lambda_{n+j}\geq 0
\end{cases}
\end{equation}
Now we claim that $\lambda^*=[\lambda_0^*,\lambda_1^*,\dots,\lambda_{2n}^*]^T$ defined by the following values is in the feasible set:
\begin{equation}
\begin{cases}
\lambda_0^*=0\\
\lambda_1^*=\dots=\lambda_n^*=0\\
\lambda_{n+i}^*=1-\delta\\
\lambda_{n+j}=\beta \text{  for } j\in N_i\\
\lambda_{n+j}=0 \text{  for } j\not\in N_i
\end{cases}
\end{equation}
We verify the claim for all possible cases as the following. Assume $\left\vert{N_i\cap I(t)}\right\vert=m$ .
\begin{multline}
\text{For $X_i=0$ : }\mathbb{P}\left( Y_i=1 \mid X\right)=1-(1-\beta)^m \leq\\
m\beta = \lambda_0^* + \sum\limits_{k=1}^n B_{X,k}\lambda_k^* + \sum\limits_{k=n+1}^{2n} B_{X,k}\lambda_k^* .
\end{multline}
\begin{multline}
\text{For $X_i=1$ : }\mathbb{P}\left( Y_i=1 \mid X\right)=1-\delta \leq\\
1-\delta+m\beta = \lambda_0^* + \sum\limits_{k=1}^n B_{X,k}\lambda_k^* + \sum\limits_{k=n+1}^{2n} B_{X,k}\lambda_k^* .
\end{multline}
\begin{multline}
\text{For $X_i=2$ : }\mathbb{P}\left( Y_i=1 \mid X\right)=0 \leq\\
m\beta = \lambda_0^* + \sum\limits_{k=1}^n B_{X,k}\lambda_k^* + \sum\limits_{k=n+1}^{2n} B_{X,k}\lambda_k^* .
\end{multline}
It follows that $\lambda^*$ is in the feasible set. Back to the Eq. \eqref{minmax} we have:
\begin{align}
\max_{\mu B=p^T, \mu\succeq 0} p_i(t+1) &= \min_\lambda \max_{\mu\succeq 0} \mu(SBf_i-B\lambda) +p^T\lambda\notag\\
 &\leq p^T\lambda^* = (1-\delta)p_i +\beta\sum\limits_{j\in N_i}p_j ,\notag
\end{align}
which proves:
\begin{equation}
p_{I,i}(t+1) \leq (1-\delta)p_{I,i}(t) + \beta\sum\limits_{j\in N_i}p_{I,j}(t) .
\end{equation}
Moreover, we already know that $p_{R,i}(t+1)=(1-\gamma)p_{R,i}(t)+\delta p_{I,i}(t)$ (Eq. \ref{exact_R}), and all the equations can be expressed in a vector form, using $p_R=[p_{R,1}(t),p_{R,2}(t),\dots,p_{R,n}(t)]^T$ and $p_I=[p_{I,1}(t),p_{I,2}(t),\dots,p_{I,n}(t)]^T$:
\begin{align}
\begin{bmatrix}p_R\\p_I\end{bmatrix}(t+1)&\preceq\begin{bmatrix}
(1-\gamma)I_n & \delta I_n\\
0_n & (1-\delta)I_n+\beta A
\end{bmatrix} \begin{bmatrix}p_R\\p_I\end{bmatrix}(t)\\
&= M \begin{bmatrix}p_R\\p_I\end{bmatrix}(t). \notag
\end{align}

The mixing time as defined before in \eqref{Eq:mixingdefn} is
$$
t_{mix}(\epsilon)=\min\left\{t : \sup_\mu \|\mu S^t-\pi\|_{TV} \leq \epsilon\right\} ,
$$
and we have:
\begin{align}
\|\mu S^t-\pi\|_{TV} &= \frac{1}{2}\sum\limits_X \lvert(\mu S^t)_X - \pi_X\rvert\\
&= \frac{1}{2}\sum\limits_X \lvert(\mu S^t)_X - (e_{\bar{0}})_X\rvert\\
&= \frac{1}{2}\big( 1-(\mu S^t)_{\bar{0}}\big) + \frac{1}{2}\sum\limits_{X\neq \bar{0}} (\mu S^t)_X\\
&= \frac{1}{2}\big( 1-(\mu S^t)_{\bar{0}}\big) + \frac{1}{2}\big( 1-(\mu S^t)_{\bar{0}}\big)\\
&= 1-(\mu S^t)_{\bar{0}}\\
&= 1-\mu S^t e_{\bar{0}}^T\\
&\leq 1-e_{\bar{1}} S^t e_{\bar{0}}^T  .
\end{align}
Hence, for any $t< t_{mix}(\epsilon)$:
\begin{align}
\epsilon &< 1-\mathbb{P}\left( \substack{\text{all nodes are susceptible at time $t$} \mid \\\text{all nodes were infected at time $0$}}\right)\\
 &=\mathbb{P}\left( \substack{\text{some nodes are infected or recovered at time $t$} \mid \\\text{all nodes were infected at time $0$}} \right)\\
 &\leq \sum\limits_{i=1}^n (p_{I,i}(t)+p_{R,i}(t)) = 1_{2n}^T \begin{bmatrix}p_R\\p_I\end{bmatrix}(t)\\
 &\leq 1_{2n}^T M^t \begin{bmatrix}p_R\\p_I\end{bmatrix}(0)\\
 &\leq 1_{2n}^T M^t 1_{2n}\\
 &\leq \|1_{2n}\|^2 \|M\|^t\\
 &=2n\|M\|^t .
\end{align}
$\|M\|<1$ leads to the fact that $t< \frac{\log \frac{2n}{\epsilon}}{-\log \|M\|}$ for all $t< t_{mix}(\epsilon)$. Therefore $t_{mix}(\epsilon)\leq \frac{\log \frac{2n}{\epsilon}}{-\log \|M\|}$, which means the mixing time is $O(\log n)$.
\hfill$\blacksquare$

{\noindent \bfseries Proof of Proposition \ref{prop_id}.}
$M$ is in fact the Jacobian matrix of the nonlinear map, and its largest eigenvalue is less than $1$ if the largest eigenvalue of $(1-\delta)I_n+P_S^*\beta A)$ is less than $1$. It follows that the fixed point is locally stable under this condition, and the statement a is true.

Eq. \eqref{nonlinear_I_id} can be upper bounded as:
\begin{align}
P&_{I,i}(t+1) =(1-\delta)P_{I,i}(t)+\notag\\
&\Big(1-\prod_{j\in N_i} (1-\beta P_{I,j}(t))\Big)(1-P_{R,i}(t)-P_{I,i}(t))\notag\\
&\leq (1-\delta)P_{I,i}(t)+ \big(\beta\sum\limits_{j\in N_i}P_{I,j}\big)(1-P_{R,i}(t)-P_{I,i}(t))\label{upperbound_1_id}\\
&\leq (1-\delta)P_{I,i}(t)+ \beta\sum\limits_{j\in N_i}P_{I,j}\label{upperbound_2_id} ,
\end{align}
which implies the statement b.

We remark that from \eqref{upperbound_1_id} to \eqref{upperbound_2_id} it is not possible to show an upper bound of $(1-\delta)P_{I,i}(t)+ \big(\beta\sum\limits_{j\in N_i}P_{I,j}\big)(1-P_R^*)$ instead; as it requires $P_{R,i}(t)+P_{I,i}(t)\geq P_R^*$, which is equivalent to $P_{S,i}(t)\leq P_S^*$, that is not true in general. 
\hfill$\blacksquare$

{\noindent \bfseries Proof of Theorem \ref{thm:SIV_id_nontrivial}.}
In the same way as in Lemma \ref{lm:v}, when $\lambda_{\max}((1-\delta)I_n+ P_S^*\beta A)>1$, there exists $v\succ0_n$ such that $(\beta A-\frac{\delta}{P_S^*} I_n)v\succ0_n$.

On the other hand, $\Psi(\begin{bmatrix}P_R\\0_n\end{bmatrix})=0_n$ and the Jacobian of $\Psi$ at the main fixed point is equal to $\begin{bmatrix}0_{n\times n} & \beta A-\frac{\delta}{1-P_R^*} I_n\end{bmatrix}_{n\times 2n}$. As a result, there exists a small $\epsilon>0$ such that $\Psi(\begin{bmatrix}\epsilon u\\ \epsilon v\end{bmatrix})=(\beta A-\frac{\delta}{P_S^*} I_n)v\epsilon$, which is $\succ 0_n$, and indicates that $\epsilon v \in U$.

The rest of the proof is the same as that of \ref{thm:SIRS_nontrivial}, with the main difference that instead of Eq. \eqref{relation}, we have the following relation:
\begin{equation}
p_{R,i}^*=\frac{\theta}{\gamma+\theta}+\frac{\delta-\theta-\delta\theta}{\gamma+\theta}p_{I,i}^*
\end{equation}
\hfill$\blacksquare$

{\noindent \bfseries Proof of Theorem \ref{thm_mixing_id}.}
The proof is similar to that of Theorem \ref{thm_mixing}, except we have $(SBf_i)_X =$
$$
\begin{cases}
(1-\beta)^{m_i}\theta,& \text{if } i\in\left\{1,2,\dots,n\right\} \text{ and } X_i=0\\
\delta,& \text{if } i\in\left\{1,2,\dots,n\right\} \text{ and } X_i=1\\
1-\gamma,& \text{if } i\in\left\{1,2,\dots,n\right\} \text{ and } X_i=2\\
1-(1-\beta)^{m_{i-n}},& \text{if } i\in\left\{n+1,\dots,2n\right\} \text{ and } X_{i-n}=0\\
1-\delta,& \text{if } i\in\left\{n+1,\dots,2n\right\} \text{ and } X_{i-n}=1\\ \label{SBf_i_id}
0,& \text{if } i\in\left\{n+1,\dots,2n\right\} \text{ and } X_{i-n}=2\\
\end{cases} \label{Bl_id}
$$
Since we are interested to evaluate $p_i(t+1)$ only for $i\in\left\{n+1,n+2,\dots,2n\right\}$, and the corresponding terms in \eqref{SBf_i_id} (the lower three) do not depend on $\theta$, the equations for optimal Lagrange multipliers are the same as in Theorem \ref{thm_mixing}. It follows that
\begin{equation}
p_I(t+1)\preceq ((1-\delta)I_n+\beta A) p_I(t) .
\end{equation}

Now for any $t< t_{mix}(\epsilon)$:
\begin{align}
\epsilon &< \mathbb{P}\left( \substack{\text{some nodes are infected at time $t$} \mid \\\text{all nodes were infected at time $0$}} \right)\\
 &\leq \sum\limits_{i=1}^n p_{I,i}(t) = 1_{n}^T p_I(t)\\
 &\leq 1_{n}^T ((1-\delta)I_n+\beta A)^t p_I(0)\\
 &\leq 1_{n}^T ((1-\delta)I_n+\beta A)^t 1_{n}\\
 &\leq \|1_{n}\|^2 \|(1-\delta)I_n+\beta A\|^t\\
 &=n\|(1-\delta)I_n+\beta A\|^t .
\end{align}
$\|(1-\delta)I_n+\beta A\|<1$ leads to $t< \frac{\log \frac{n}{\epsilon}}{-\log \|(1-\delta)I_n+\beta A\|}$ for all $t< t_{mix}(\epsilon)$, and therefore $t_{mix}(\epsilon)\leq \frac{\log \frac{n}{\epsilon}}{-\log \|(1-\delta)I_n+\beta A\|}$.
\hfill$\blacksquare$

{\noindent \bfseries Proof of Proposition \ref{prop_vd}.}
The statement a is again clear since if the largest eigenvalue of $(1-\delta)I_n+(1-\theta) P_S^*\beta A$ is less than one, then the largest eigenvalue of $M$ is less than $1$, which means the norm of the Jacobian matrix is less than $1$.

The statement b also follows from upper bounding Eq. \eqref{nonlinear_I_vd} as
\begin{equation}
P_{I,i}(t+1) \leq (1-\delta)P_{I,i}(t)+ (1-\theta)\beta\sum\limits_{j\in N_i}P_{I,j}\label{upperbound_2_vd} .
\end{equation}
\hfill$\blacksquare$

{\noindent \bfseries Proof of Theorem \ref{thm_mixing_vd}.}
We use the same linear programming argument as in the proofs of Theorems \ref{thm_mixing} and \ref{thm_mixing_id}, and show that for each $i\in\left\{n+1,n+2,\dots,2n\right\}$, we have $p_i(t+1) \leq (1-\delta)p_i(t) + (1-\theta)\beta\sum\limits_{j\in N_i}p_j(t)$. The main difference is $(SBf_i)_X =$
$$
\begin{cases}
\theta,& \text{if } i\in\left\{1,2,\dots,n\right\} \text{ and } X_i=0\\
\delta,& \text{if } i\in\left\{1,2,\dots,n\right\} \text{ and } X_i=1\\
1-\gamma,& \text{if } i\in\left\{1,2,\dots,n\right\} \text{ and } X_i=2\\
\!\begin{aligned}(1-\theta)(1-(1\ \ \\-\beta)^{m_{i-n}}),\end{aligned}& \text{if } i\in\left\{n+1,\dots,2n\right\} \text{ and } X_{i-n}=0\\
1-\delta,& \text{if } i\in\left\{n+1,\dots,2n\right\} \text{ and } X_{i-n}=1\\ \label{SBf_i_vd}
0,& \text{if } i\in\left\{n+1,\dots,2n\right\} \text{ and } X_{i-n}=2\\
\end{cases}.
$$
It can be verified that the Lagrange multiplier vector $\lambda^*=[\lambda_0^*,\lambda_1^*,\dots,\lambda_{2n}^*]^T$ with the following values is in the feasible set:
\begin{equation}
\begin{cases}
\lambda_0^*=0\\
\lambda_1^*=\dots=\lambda_n^*=0\\
\lambda_{n+i}^*=1-\delta\\
\lambda_{n+j}=\beta(1-\theta) \text{  for } j\in N_i\\
\lambda_{n+j}=0 \text{  for } j\not\in N_i
\end{cases}
\end{equation}
and it leads to
\begin{equation}
p_I(t+1)\preceq ((1-\delta)I_n+\beta(1-\theta) A) p_I(t) .
\end{equation}
Under the condition that $\frac{\beta(1-\theta)\lambda_{\max}(A)}{\delta}<1$, by the same argument as in the proof of Theorem \ref{thm_mixing_id}, $t_{mix}(\epsilon)\leq \frac{\log \frac{n}{\epsilon}}{-\log \|(1-\delta)I_n+\beta(1-\theta) A\|}$.
\hfill$\blacksquare$



\ifCLASSOPTIONcompsoc
  \section*{Acknowledgments}
\else
  \section*{Acknowledgment}
\fi

This work was supported in part by the National Science Foundation under grants CNS-0932428, CCF-1018927, CCF-1423663 and CCF-1409204, by a grant from Qualcomm Inc., by NASA's Jet Propulsion Laboratory through the President and Director's Fund, and by King Abdullah University of Science and Technology.
The authors would like to thank Christos Thrampoulidis, Ehsan Abbasi, Ramya K. Vinayak, Matthew D. Thill, Wei Mao and Subhonmesh Bose for many insightful discussions on the subject.

\ifCLASSOPTIONcaptionsoff
  \newpage
\fi



\bibliographystyle{IEEEtran}
\bibliography{IEEEabrv,references}
\end{document}